\documentclass[runningheads]{amsart}

\usepackage{amsmath,amssymb,amsthm}
 \usepackage{amsaddr}
\usepackage{xspace}
\usepackage[disable]{todonotes}
\usepackage{virginialake}
\usepackage[colorlinks]{hyperref}
\usepackage[capitalise]{cleveref}
\usepackage{tikz}
\usetikzlibrary{arrows,cd,tikzmark}

\usepackage{thread}
\usepackage{transparent}

\begin{document}

\renewcommand{\phi}{\varphi}
\renewcommand{\emptyset}{\varnothing}
\renewcommand{\epsilon}{\varepsilon}
\newcommand{\defname}[1]{\textbf{{#1}}}
\newcommand{\df}{:=}
\newcommand{\bnf}{::=}
\newcommand{\Pow}{\mathcal P}
\newcommand{\pow}[1]{\Pow(#1)}

\newcommand{\IH}{\mathit{IH}}

\newcommand{\Nat}{\mathbb{N}}
\newcommand{\Int}{\mathbb{Z}}
\newcommand{\Cal}[1]{\mathcal{#1}}
\newcommand{\SF}[1]{\mathsf{#1}}
\newcommand{\BB}[1]{\mathbb{#1}}
\newcommand{\resp}{\textit{resp.}\xspace}
\newcommand{\aka}{\textit{aka}\xspace}
\newcommand{\viz}{\textit{viz.}\xspace}
\newcommand{\ie}{\textit{i.e.}\xspace}
\newcommand{\cf}{\textit{cf.}\xspace}
\newcommand{\wrt}{\textit{wrt}\xspace}
\newcommand{\rk}[1]{\SF{rk}(#1)}

\newcommand{\red}[1]{{\color{red} #1}}
\newcommand{\blue}[1]{{\color{blue} #1}}
\newcommand{\mgt}[1]{\textcolor{magenta}{#1}}
\newcommand{\cyn}[1]{\textcolor{cyan}{#1}}
\newcommand{\olv}[1]{\textcolor{olive}{#1}}
\newcommand{\orange}[1]{{\color{orange}#1}}
\newcommand{\purple}[1]{{\color{purple}#1}}
\newcommand{\green}[1]{{\color{ForestGreen}#1}}

\newcommand{\anupam}[1]{\todo{Anu: #1}}
\newcommand{\abhishek}[1]{\todo{Abhi: #1}}
\newcommand{\note}[1]{\mgt{#1}}


\newtheorem{theorem}{Theorem}
\newtheorem{conjecture}{Conjecture}
\newtheorem{lemma}[theorem]{Lemma}
\newtheorem{proposition}[theorem]{Proposition}
\newtheorem{observation}[theorem]{Observation}
\crefformat{observation}{Observation~#2#1#3}

\newtheorem{corollary}[theorem]{Corollary}
\newtheorem{fact}[theorem]{Fact}

\theoremstyle{definition}
\newtheorem{definition}[theorem]{Definition}
\newtheorem{example}[theorem]{Example}
\newtheorem{remark}[theorem]{Remark}
\newtheorem{convention}[theorem]{Convention}

\newcommand{\Alphabet}{\mathcal{A}}
\newcommand{\Var}{\mathcal{V}}

\newcommand{\proves}{\vdash}

\newcommand{\Lang}{\mathcal L}
\newcommand{\lang}[1]{\Lang(#1)}

\newcommand{\wLang}{\Lang}
\newcommand{\wlang}[1]{\wLang(#1)}

\newcommand{\FV}{\mathrm{FV}}
\newcommand{\fv}[1]{\FV(#1)}

\newcommand{\fint}[1]{\lceil #1 \rceil}

\newcommand{\FL}{\mathrm{FL}}
\newcommand{\flred}{\rightarrow_{\FL}}
\newcommand{\flredeq}{\flred^{=}}
\newcommand{\fl}[1]{\FL(#1)}
\newcommand{\eqfl}{=_\FL}
\newcommand{\lefl}{<_\FL}
\newcommand{\leqfl}{\leq_\FL}
\newcommand{\geqfl}{\geq_\FL}

\newcommand{\subform}{\sqsubseteq}
\newcommand{\supform}{\sqsupseteq}
\newcommand{\dle}{\prec}
\newcommand{\dleq}{\preceq}
\newcommand{\dge}{\succ}
\newcommand{\dgeq}{\succeq}


\newcommand{\LTL}{\mathsf{LTL}}
\newcommand{\muLTL}{\mu\LTL}

\newcommand{\seqar}{\rightarrow}

\newcommand{\id}{\mathsf{id}}

\newcommand{\K}{\mathsf{h}}
\newcommand{\kk}[1]{\K_{#1}}
\newcommand{\wk}{\mathsf{w}}
\newcommand{\cntr}{\mathsf{c}}
\newcommand{\cut}{\mathsf{cut}}

\newcommand{\partition}{\mathsf p}
\newcommand{\dis}{\lr \partition}
\newcommand{\exh}{\rr \partition}

\newcommand{\lr}[1]{#1\text{-}l}
\newcommand{\rr}[1]{#1\text{-}r}
\newcommand{\func}{\SF{func}}

\newcommand{\KA}{\mathsf{KA}}
\newcommand{\lhKA}{\ell\KA}
\newcommand{\HKA}{\mathsf{HKA}}

\newcommand{\RLA}{\ensuremath{\mathsf{RLA}}\xspace}
\newcommand{\RLL}{\ensuremath{\mathsf{RLL}}\xspace}
\newcommand{\fd}[1]{#1^{\mathit{fd}}}
\newcommand{\RLLfd}{\fd\RLL}

\newcommand{\RLLLang}{\RLL_\Lang}

\newcommand{\ubdd}{\mathsf u}
\newcommand{\gdd}{\mathsf g}
\newcommand{\gRLL}{\gdd\RLL}
\newcommand{\gRLLfd}{\fd\gRLL}
\newcommand{\uRLL}{\ubdd\RLL}
\newcommand{\uRLLfd}{\fd\uRLL}

\newcommand{\CRLL}{\ensuremath{\mathsf{CRLL}}\xspace}
\newcommand{\LRLA}{\ensuremath{\mathsf{LRLA}}\xspace}
\newcommand{\LRLL}{\ensuremath{\mathsf{LRLL}}\xspace}
\newcommand{\LRLAhat}{\widehat\LRLA}
\newcommand{\LRLLhat}{\widehat\LRLL}
\newcommand{\JRLAhat}{\widehat{\mathsf{JRLA}}}
\newcommand{\CRLA}{\ensuremath{\mathsf{CRLA}}\xspace}

\newcommand{\CRLLfd}{\CRLL^{\mathit{fd}}}

\newcommand{\uCRLL}{\ubdd\CRLL}
\newcommand{\uLRLLhat}{\ubdd\LRLLhat}
\newcommand{\uCRLLfd}{\fd\uCRLL}
\newcommand{\uLRLLhatfd}{\fd\uLRLLhat}

\newcommand{\LRLLLanghat}{\mathsf L\widehat\RLL_\Lang}
\newcommand{\CRLLLang}{\CRLL_\Lang}
\newcommand{\LRLLLang}{\mathsf L \RLL_\Lang}

\newcommand{\ind}{\mathsf{ind}}
\newcommand{\coind}{\mathsf{coind}}

\newcommand{\CRLLLangc}{\CRLLLang^c}
\newcommand{\LRLLLangc}{\LRLLLang^c}


\newcommand{\cons}{::}
\newcommand{\inl}[1]{\SF{inl}(#1)}
\newcommand{\inr}[1]{\SF{inr}(#1)}
\newcommand{\terms}[1]{\left\lfloor #1 \right\rfloor}
\newcommand{\restrict}[2]{{
  \left.\kern-\nulldelimiterspace 
  #1 
  \littletaller 
  \right|_{#2} 
  }}
\newcommand{\littletaller}{\mathchoice{\vphantom{\big|}}{}{}{}}
\newcommand{\stub}[2]{\vltr {#1}{#2}{\vlhy{}}{\vlhy{}}{\vlhy{}}}


\newcommand{\cutred}{\triangleright}


\renewcommand{\c}[1]{#1^c}
\newcommand{\infrule}{\mathsf{r}}

\newcommand{\strat}[1]{\mathfrak{#1}}

\newcommand{\gtop}{\top_\Alphabet}


\newcommand{\Eloise}{\exists}
\newcommand{\Abelard}{\forall}

\newcommand{\prover}{\mathbf{P}}
\newcommand{\denier}{\mathbf{D}}

\newcommand{\cantor}{\mathcal C}

\title[Cyclic system for an algebraic theory of APAs]{Cyclic system for an algebraic theory of alternating parity automata}

\author{Anupam Das \and Abhishek De}
\address{University of Birmingham}
\date{\today}

\begin{abstract}
$\omega$-regular languages are a natural extension of the regular languages to the setting of infinite words. Likewise, they are recognised by a host of automata models, one of the most important being Alternating Parity Automata (APAs), a generalisation of Büchi automata that symmetrises both the transitions (with universal as well as existential branching) and the acceptance condition (by a parity condition). 

In this work we develop a cyclic proof system manipulating APAs, represented by an algebraic notation of Right Linear Lattice expressions. This syntax dualises that of previously introduced Right Linear Algebras, which comprised a notation for non-deterministic finite automata (NFAs). This dualisation induces a symmetry in the proof systems we design, with lattice operations behaving dually on each side of the sequent. 
Our main result is the soundness and completeness of our system for $\omega$-language inclusion, heavily exploiting game theoretic techniques from the theory of $\omega$-regular languages.
\end{abstract}

\maketitle    

\todo[inline]{General todos:



* add examples/development of complementation after completeness

* Add translation from APAs to RLL expressions (possibly in appendices)

* discuss complexity of proof search

* improve completeness to unguarded sequents
}

\section{Introduction}

The theory of $\omega$-regular languages is among the most important in all of computer science. 
\emph{B\"uchi automata}, the classical characterisation of $\omega$-regularity, are a variation of usual finite state automata that run on infinite words.
B\"uchi's famous complementation theorem for these automata is the engine underlying his proof of the decidability of monadic second-order logic (MSOL) over infinite words \cite{Buchi66:buchi-compl-mso}.
Its extension to infinite trees, \emph{Rabin's Tree Theorem} \cite{Rabin68}, is often referred to as the `mother of all decidability results'.

McNaughton showed that, while B\"uchi automata could not be determinised per se, a naturally larger class of acceptance conditions (Muller or parity) allowed such determinisation, a highly technical result later improved by Safra \cite{Safra}. 
A later relaxation was the symmetrisation of the transition relation itself: instead of only allowing nondeterministic states, i.e.\ existential branching, allow \emph{co-nondeterministic} ones too, i.e.\ universal branching. 
This has led to beautiful accounts of $\omega$-language theory via the theory of positional and finite memory games;
the resulting computational model, \emph{alternating parity automata} (APAs), 
is now the go-to model in textbook presentations of $\omega$-regular language theory \cite{GTW03,PerPin:inf-word-aut-book,Bojan23course}. 
Indeed, their features more closely mimic those of logical settings where such symmetries abound, e.g. the linear-time $\mu$-calculus ($\muLTL$) and MSOL over infinite words.

\subsection{Contribution}
In this work we design a system for reasoning \emph{natively} about APAs, in the form of \emph{right-linear lattice} (RLL) expressions. 
This syntax is a dualisation of previously studied \emph{right-linear algebra} (RLA) expressions \cite{DD24a,DasDe24:rla-preprint}, which comprise a notation for non-deterministic finite word automata (NFAs). 
While RLA expressions model non-determinism by a join-semilattice structure $(0,+)$, and resolve cycles of an automaton by \emph{least} fixed points $\mu$, RLL expressions can further model co-nondeterminism by a lattice structure $(0,+,\top,\cap)$ and resolve infinite paths of an automaton by a combination of least and \emph{greatest} fixed points $\nu$, modelling the parity condition.

Our system, $\CRLLLang$, is a two-sided sequent calculus admitting \emph{cyclic proofs}: proof trees may be non-wellfounded but regular. 
Thus, while usual inductive proofs may be represented as finite trees or dags, cyclic proofs are represented by finite graphs, possibly with cycles.
Naturally non-wellfounded reasoning may be fallacious, and so $\CRLLLang$ is equipped with a global correctness condition at the level of infinite paths along formula ancestry.
Our main result is the soundness and completeness of $\CRLLLang$ with respect to the intended model of $\omega$-languages. 
Thus $\CRLLLang$ is suitable for reasoning directly about APA inclusions (and thereby emptiness, universality, equivalence).

The techniques we employ in this work draw heavily from the literature on cyclic proofs, in particular the \emph{game theoretic} approach dating back to Niwinsk\'i and Walukiewicz \cite{NiwWal96:games-mu-calc}.
Soundness of $\CRLLLang$ is established via an infinite descent argument that is typical of cyclic proof theory. 
In this work we formulate this argument as a reduction to the adequacy of certain \emph{evaluation games}, mirroring the well-known \emph{acceptance games} for APAs.
Completeness of $\CRLLLang$ further exploits the finite-memory determinacy of the associated \emph{proof search game}, implied by the well-known B\"uchi-Landweber theorem for $\omega$-regular games \cite{BuchiLandweber}.

\subsection{Related work}
We have already mentioned the previous work \cite{DD24a,DasDe24:rla-preprint} that studied RLA expressions notating NFAs. 
That same work also considered an extension by greatest fixed points, but without meet-lattice structure $(\top,\cap)$. 
Such expressions thus model \emph{nondeterministic} parity automata, so the corresponding system $\nu\CRLA$ is strictly subsumed by the system $\CRLLLang$ we present here.
In particular the symmetry of APAs renders $\CRLLLang$ more symmetric than $\nu\CRLA$, with lattice operations behaving dually on each side of a sequent.

Using fixed points to model parity conditions is an old idea, going back to work of Streett and Emerson \cite{StrEme89:aut-th-proc-mu-calc} in the setting of the modal $\mu$-calculus.
The latter's {linear-time} restriction $\muLTL$ offers an alternative syntax for APAs, but one that we argue is not as close as RLL expressions. 
In particular $\muLTL$ formulas, built over classical logic, are equipped with native complementation, while RLL expressions are not: they are set in the language of lattice theory rather than Boolean algebra.
Previously studied systems for $\muLTL$ include a complete axiomatisation due to Kaivola \cite{Kaivola95} (see also \cite{Doumane17}) and, most related to this work, a cyclic system due to Dax, Hofmann and Lange \cite{DHL06:muLTL}.

Our algebraic syntax is rather inspired by the tradition of \emph{Kleene algebras}, certain structures interpreting regular expressions, and friends (see, e.g., \cite{Salomaa66,Krob91,Kozen94,Boffa90,Boffa95}).
Regular expressions can be viewed as a notation for NFAs, though part of the motivation for RLA expressions in \cite{DD24a,DasDe24:rla-preprint} was to study a more native syntax, by adding multiplication but accommodating fixed points.
In fact the \emph{Right Linear Algebras} proposed in that work strictly generalise (even left-handed) Kleene algebras.
\emph{$\omega$-regular expressions} are a modification of regular expressions that similarly model B\"uchi automata.
They have enjoyed both axiomatic treatments (see, e.g., \cite{Wagner76,Cohen00,LS12,CLS15}) and more recently a cyclic system \cite{HazKup22}, building on the cyclic system of \cite{DasPou17:hka} for Kleene Algebra.

\subsection{Structure of paper}
In \cref{sec:prelims} we present RLL expressions as a notation for APAs, along with their intended $\omega$-language semantics and give several examples.
In \cref{sec:eval-games} we define evaluation games for RLL expressions, and prove their adequacy for the language semantics.
In \cref{sec:crll} we define our cyclic system $\CRLLLang$ and give several examples of (non-)proofs.
We prove the soundness and completeness of $\CRLLLang$ in \cref{sec:adequacy}, eventually by reduction to the adequacy of evaluation games.
Finally we present some concluding remarks in \cref{sec:concs}.

\subsection{Background and prerequisites}
We do not formally assume any particular prerequisites, and aim to present results in as self-contained a manner as possible. 
Nonetheless it is helpful to have some familiarity with the theory of $\omega$-automata and their correspondences with logics and games.
Useful references include the books \cite{GTW03,PerPin:inf-word-aut-book}, as well as the course \cite{Bojan23course}, from which we take our basic definitions.
Naturally any familiarity with our previous work \cite{DD24a,DasDe24:rla-preprint} would be useful, but not strictly necessary.
\section{RLL expressions: a notation for alternating parity automata}
\label{sec:prelims}

Let us fix a finite set $\Alphabet$ (the \defname{alphabet}) of \defname{letters}, written $a,b,$ etc., and a countable set $\Var$ of \defname{variables}, written $X,Y,$ etc.

\subsection{Syntax and semantics of right-linear lattice expressions}
In this subsection we introduce the basic syntax and semantics we shall work with, before relating them to automaton models later.

\begin{definition}
\defname{Right-linear lattice expressions}, or simply \textbf{(RLL-)expressions}, written $e,f$, etc. are generated as follows,
\[
\begin{array}{rcl@{\quad \mid \quad }c@{\quad \mid \quad }c@{\quad \mid \quad }c}
     e,f,\dots & \bnf \quad & X \quad \mid \quad ae & 0 & e+f & \mu X e \\
        &   &   & \top & e\cap f & \nu X e
\end{array}
\] 
where $X\in\Var$ and $a\in\Alphabet$. 

The \textbf{free variables} of an expression are defined as expected, understanding $\mu$ and $\nu$ as variable binders.
A \defname{closed} expression is one with no free variables (otherwise it is \defname{open}).
A (closed) expression $e$ is \defname{guarded} if each of its variable occurrences $X$ occurs free in a subexpression of form $af$.
\end{definition}

We may sometimes refer to expressions as `formulas' when it is more natural (e.g.\ `subformula', or `principal formula').
The intended semantics of expressions is given by languages of infinite words: 
\begin{definition}
    [Language semantics]
    \label{dfn:lang-semantics}
    Let us temporarily expand the syntax of expressions to include each language $A\subseteq \Alphabet^\omega$ as a constant symbol. We interpret each closed expression (of this expanded language) as a subset of $ \Alphabet^\omega$ inductively as follows:
    \[
    \begin{array}{r@{\ \df \ }l}
         \lang A & A \\
         \lang 0 & \emptyset \\
         \lang{e+f} & \lang e \cup \lang f \\
         \lang{\mu X e(X)} & \bigcap \{A \supseteq \lang{e(A)}
    \end{array}
    \qquad
    \begin{array}{r@{\ \df \ }l}
         \lang {ae} & \{aw : w \in \lang e\} \\
         \lang \top & \Alphabet^\omega \\
         \lang{e\cap f} & \lang e \cap \lang f \\
         \lang{\nu X e(X)} & \bigcup \{A \subseteq \lang{e(A)}
    \end{array}
    \]
\end{definition}

The interpretation of $a\cdot, 0,\top, +, \cap$ should be familiar from formal language theory. 
For the binders, the idea is that $\Lang$ interprets $\mu X e(X)$ and $\nu X e(X)$ as the least and greatest fixed points, respectively, of the operation $ A\mapsto \lang{e(A)}$.
It is not hard to see that each such operation is \emph{monotone}, i.e.\ $A\subseteq B \implies \lang {e(A)} \subseteq \lang {e(B)}$, by a straightforward induction on the structure of $e$.\todo{could insert statement and proof}
Thus by the \emph{Knaster-Tarski Theorem} $ A\mapsto \lang{e(A)}$ indeed has a least and greatest fixed point in $\Lang$, given by the intersection of pre-fixed points and union of post-fixed points, motivating the definitions of $\lang{\mu X e(X)}$ and $\lang{\nu X e(X)}$.

\begin{example}
[Empty and universal languages]
\label{ex:empty-and-universal}
We have $\lang{\mu X X}=\emptyset$ and $\lang{\nu X X }=\Alphabet^\omega$.
Thus we can say that the structure $\Lang$ satisfies $0= \mu X X$ and $\top = \nu X X$.
\end{example}

\begin{example}
    [$\omega$-iteration]
    \label{ex:omega-iteration}
    We have $\lang {\nu X (aX)} = a^\omega$ and $\lang {\nu X (aX + bX)} = \{a,b\}^\omega$.
\end{example}

In previous work \cite{DD24a,DasDe24:rla-preprint} we studied expressions without $\top, \cap$ and $\nu$, called \emph{right-linear algebra (RLA) expressions}, which semantically denote languages of finite words rather than $\omega$-words.\footnote{More precisely, one must include a constant symbol $1$, with $\lang 1 \df \{\epsilon\}$, so that expressions do not trivially always denote the empty language.} 
The terminology `right-linear' is drawn from the context-free grammar literature, as both RLA expressions and RLL expressions allow products $ef$ only when $e$ is a letter $a \in \Alphabet$.
RLA expressions can be construed as a notation for non-deterministic finite automata (equivalently, right-linear grammars), and duly denote just the regular languages.

As RLL expressions denote languages of infinite words, we are interested in the corresponding notion of regularity.
Let us henceforth freely use operations from formal language theory when manipulating languages, e.g.\ writing $A^*$, $A^\omega$ and $AB$ for their usual definitions, when $A\subseteq \Alphabet^*$ and $B\subseteq \Alphabet^{\leq\omega}$. 

\begin{definition}
[$\omega$-regular languages]
    A language $A \subseteq \Alphabet^\omega$ is \defname{$\omega$-regular} if we have $A = \bigcup\limits_{i<n} B_iC_i^\omega$ for some $B_i,C_i$ regular and $\epsilon \notin C_i$.\footnote{In fact the requirement that $\epsilon \notin C_i$ can be safely dropped, understanding that $\epsilon^\omega$ be identified with $\Alphabet^\omega$.} 
\end{definition}

There are now several equivalent presentations of $\omega$-regular languages. 
The one above is often called the \emph{Kleene closure} of regular languages, a general way to define the infinite-word analogue of a class of finite-word languages.
The most common presentations are via $\omega$-automata, such as \emph{B\"uchi automata} and \emph{parity automata} (see, e.g., \cite{GTW03,PerPin:inf-word-aut-book,Bojan23course}); we shall not survey these in detail here, but will develop \emph{alternating} parity automata later in the section.

It turns out that RLL expressions denote just the $\omega$-regular languages, wrt the interpretation $\lang \cdot$. 
One direction, that RLL expressions \emph{exhaust} the $\omega$-regular languages, was already observed in previous work:

\begin{proposition}
[\cite{DD24a}]
    \label{prop:omega-reg-have-munu-exps}
For every $\omega$-regular language $A\subseteq \Alphabet^\omega$ there is a guarded expression $e$ with $A =\lang e$.
\end{proposition}

\cite{DD24a,DasDe24:rla-preprint} further studied the extension of RLA expressions by $\nu$ (but without $\top,\cap$), so-called $\nu$RLA expressions; indeed the above result holds already for $e$ free of $\top,\cap$.
The proof theory of $\nu$RLA expressions developed in \cite{DD24a,DasDe24:rla-preprint} is much simpler (but also more restricted) than that of RLL expressions presented in this work. 
Indeed the presence of $\top$ and $\cap$ renders our syntax fully symmetric: $\top$ is dual to $0$, $+$ is dual to $\cap$, and $\mu$ is dual to $\nu$. 
While $\nu$RLA expressions correspond to \emph{non-deterministic} parity automata, RLL expressions correspond to the symmetric model of \emph{alternating} parity automata.

\begin{example}
    [(In)finitely many]
    \label{ex:(in)finitely-many}
    Let us fix the alphabet $\Alphabet = \{a,b\}$.

    $f_a \df \mu X (aX + bX +  \nu Y(bY))$ denotes (in $\Lang$) the language $F_a$ of words with only finitely many $a$s. To see this let us reason within the structure $\Lang$. 
    Recalling that $\lang{\nu Y (bY)} = b^\omega$, clearly $F_a$ is a prefixed point as it is closed under $A \mapsto aA + bA +b^\omega$. To see that it is the least such point, let $A$ be another prefixed point. We have:
    \[
    \begin{array}{rcll}
         A & \supseteq & aA + bA + b^\omega \\
            & \supseteq & (a+b)A + b^\omega \\
            & \supseteq & (a+b)(a+b)A + (a+b)b^\omega \\
            & \vdots & \\
            & \supseteq & \sum\limits_{n<\omega}(a+b)^nb^\omega \\
            & \supseteq & F_a
    \end{array}
    \]

     $i_a \df \nu X \mu Y (aX + bY)$ denotes the language $I_a$ of words with infinitely many $a$s. 
     First note that, for any language $A$, we have $\lang{\mu Y (A + bY)} = b^*A$.
     From here, to show that $I_a$ is a postfixed point of $A \mapsto \mu Y (aA + bY) $, it suffices to show that $I_a \subseteq b^*aI_a$; this holds since every word $w$ with infinitely many $a$s can be written $w = b^*a w'$.
      Now suppose $B$ is another postfixed point. We have:
      \[\begin{array}{rcl}
          B & \subseteq & b^*aB \\
            & \subseteq & b^*ab^*aB \\
            & \vdots & \\
            & \subseteq & (b^*a)^\omega \\
            & \subseteq & I_a
      \end{array}\]
    
    Now, putting these together, we have that $\lang{i_a \cap f_b}$ is the language of infinite words with infinitely many $a$s but finitely many $b$s.
    It is also immediate that $f_a\cap f_b$ and $i_a\cap f_a$ both denote the empty language in $\Lang$.
\end{example}

\subsection{Expressions as automata}

In this subsection we shall introduce a textbook automaton model for computing $\omega$-languages, essentially following the exposition in \cite{Bojan23course}.

An \defname{alternating parity automaton} (APA) is a tuple $\mathbf A = (Q, \Delta, X_0, c) $ where:
\begin{itemize}
    \item $Q$ is a finite set of \defname{states}, partitioned into \defname{existential} states $E$ and \defname{universal} states $A$.
    \item $\Delta$ is a set of \defname{transitions} or \defname{productions} of form $X \to Y$ or $X\underset{a}{\to} Y$, for $X,Y \in Q$ and $a \in \Alphabet$.
    \item $X_0 \in Q$ is the \defname{initial} state.
    \item $c: Q \to \{0, \dots, n\}$ is called the \defname{colouring}.
\end{itemize}
The semantics of APAs can actually be reduced to that of RLL expressions, but let us briefly recall a self-contained definition.
A \defname{run-tree} of a word $w\in \Alphabet^\omega$ is a tree of nodes of form $(v,X)$ where $v$ is a (infinite) suffix of $w$ and $X$ is a state, generated by:
\begin{itemize}
\item The root is $(w,X_0)$.
\item A node $(v,X)$, where $X\in E$, must have exactly one child, either $(v,X)$ with $X\to Y$ in $\Delta$, or $(v',X)$ with $X\underset a \to  Y$ in $\Delta$ and $v=av'$.
    \item A node $(v,X)$, where $X\in A$, must have children $(v,Y)$ for all transitions $X \to Y$ in $\Delta$ and children $(v',Y)$ for all transitions $X\underset a \to Y$ if $v = av'$.\anupam{important subtlety: semantics of the universal quantification: for each letter or for each transition? ultimately makes no difference.}
\end{itemize}
An infinite path $\pi$ of a run-tree from $w$ is \defname{accepting} $w$ if the least colour of state occurring infinitely often along it is even (otherwise $\pi$ is \defname{rejecting}).
An APA $\mathbf A$ \defname{accepts} $w$ if there is a run-tree from $w$ in which every infinite path is accepting (otherwise $\mathbf A$ \defname{rejects} $w$).
We write $\lang {\mathbf A}$ for the set of $\omega$-words accepted by $\mathbf A$.

Notably, APAs comprise another characterisation of the $\omega$-regular languages:
\begin{theorem}
    [See, e.g., \cite{Bojan23course}]
    \label{thm:apa-omega-regular}
    Let
    $A \subseteq \Alphabet^\omega$. $A$ is $\omega$-regular $\iff$ there is an APA $\mathbf A$ with $\lang {\mathbf A} = A$.
\end{theorem}

\begin{remark}
[$\epsilon$-transitions and alternation]
    Note that we have allowed `$\epsilon$-transitions' of form $X\to Y$, in order to mimic the syntax of RLL-expressions as closesly as possible.
    While this is not a common choice, it is easy to see that it does not increase the ultimate expressivity of the model.
    Semantically, in the presence of $\epsilon$-transitions, infinite paths along run-trees are not required to exhaust the $\omega$-word being read.
    I.e.\ it is possible that an infinite path along a run-tree has first component stabilising at some particular suffix of the input, e.g.\ a path $\pi$ of a run-tree may have form $(abw,X_0), (bw, X_1), (w, X_2), (w,X_3), \dots$. 
    In this case, this path does not distinguish $abw$ from any other word $abw'$, in that $\pi$ is accepting if and only if $\pi' = (abw',X_0), (bw', X_1), (w', X_2), (w',X_3), \dots$ is accepting.
    From this point of view, for instance, we should set $\epsilon^\omega = \Alphabet^\omega$ in $\Lang$. 
    Note that this is consistent with the equation $e^\omega = \nu X (eX)$, as $\epsilon $ is the unit of concatenation.

    Another choice in our exposition is that we insist that each state is either existential or universal, rather than allowing the transition relation from a state, for each letter, to be a positive Boolean combination of states. 
The two models are easily seen to be equivalent in the presence of $\epsilon $-transitions.
The definition of APAs we have given, with $\epsilon$-transitions, existential and universal states, can be found in, e.g., \cite{Bojan23course}.
\end{remark}

We shall draw APAs in a similar fashion to usual finite word automata.

\begin{example}
[Expressions vs automata]
\label{ex:apas}
The following APAs compute the languages from \cref{ex:empty-and-universal,ex:omega-iteration}:
\[
\begin{array}{c@{}c@{}c@{}c}
\emptyset: & \Alphabet^\omega: & a^\omega: & \{a,b\}^\omega: \\
 \begin{tikzpicture}[baseline = (X)]
\node (nil) at (-1,0) {};
\node (lin) at (1,0) {};
    \node[circle, draw, label=below:{$\orange 1$}] (X) at (0,0) {$X$};
    \draw[->] (nil) edge (X);
    \draw (X) edge[loop above] (X);
\end{tikzpicture}
&
 \begin{tikzpicture}[baseline = (X)]
\node (nil) at (-1,0) {};
\node (lin) at (1,0) {};
    \node[circle, draw, label=below:{$\mgt 0$}] (X) at (0,0) {$X$};
    \draw[->] (nil) edge (X);
    \draw (X) edge[loop above] (X);
\end{tikzpicture}
&
 \begin{tikzpicture}[baseline = (X)]
\node (nil) at (-1,0) {};
\node (lin) at (1,0) {};
    \node[circle, draw, label=below:{$\mgt 0$}] (X) at (0,0) {$X$};
    \draw[->] (nil) edge (X);
    \draw (X) edge[loop above] node {$a$} (X);
\end{tikzpicture}
&
 \begin{tikzpicture}[baseline = (X)]
\node (nil) at (-1,0) {};
\node (lin) at (1,0) {};
    \node[circle, draw, label=below:{$\mgt 0$}, color=blue] (X) at (0,0) {$X$};
    \draw[->] (nil) edge (X);
    \draw (X) edge[loop above] node {$a,b$} (X);
\end{tikzpicture}
\\
\noalign{\smallskip}
\orange \mu X X & \mgt \nu X X & \mgt \nu X (aX) & \mgt\nu X (aX \blue + bX)
\end{array}
\begin{array}{l}
     \textbf{Key} \\
     \noalign{\smallskip}
     \begin{array}{r@{\ : \ }l}
     \blue \bigcirc & \text{existential state} \\
          \red \bigcirc & \text{universal state} \\
          \mgt n & \text{even colour} \\
          \orange n & \text{odd colour}
     \end{array}
\end{array}
\]
We have repeated the expressions from \cref{ex:empty-and-universal,ex:omega-iteration} for the corresponding languages above too, now with colouring suggestive of how APAs and RLL expressions correspond to each other (more on this later).
For states that are black/uncoloured, it does not matter whether they are universal or existential, as there is a unique transition from them.
We shall use the same colouring and notation conventions henceforth.

The following APAs compute the languages from \cref{ex:(in)finitely-many}:
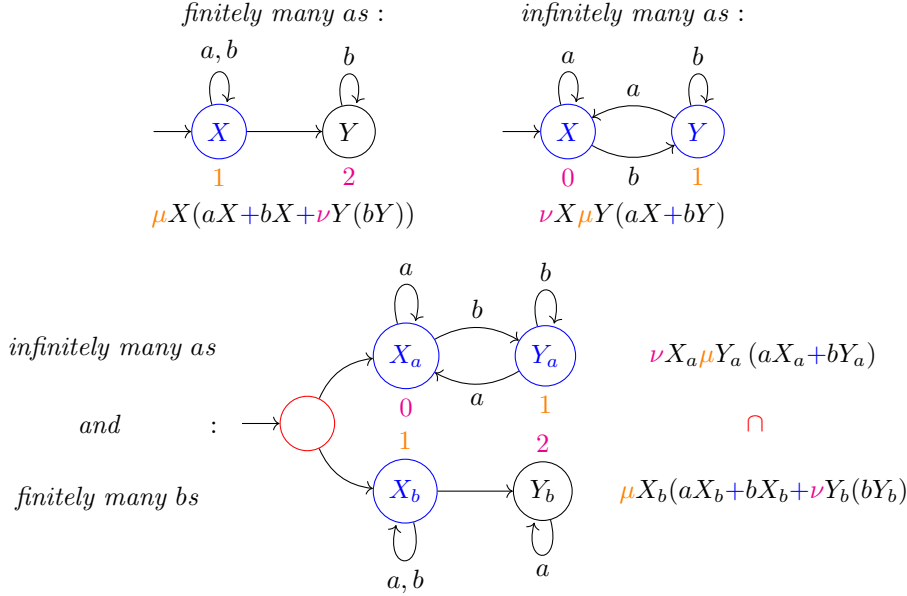
  \[
    \begin{array}{c@{\qquad}c}
            \text{\emph{finitely many $a$s}} :
& 
\text{\emph{infinitely many $a$s}} :
         \\
\begin{tikzpicture}[baseline = (X)]
    \node[circle,color=blue, draw, label=below:{$\orange 1$}] (X) at (0,0) {$X$}; 
    \node[right=1cm of X, circle, draw, label=below:{$\mgt 2$}] (Y) {$Y$};
    \node[left=.5cm of X] (nil) {};
    \node[right=.5cm of Y] (lin) {};
    \draw[->] (nil) to (X);
    \draw[->] (X) edge (Y);
    \draw (X) edge [loop above] node {$a,b$} (X);
\draw (Y) edge [loop above] node {$b$} (Y);
\end{tikzpicture}
&
\begin{tikzpicture}[baseline = (X)]
    \node[circle, color=blue, draw, label=below:{$\mgt 0$}] (X) at (0,0) {$X$}; 
    \node[right=1cm of X, circle, draw, label=below:{$\orange 1$}, color=blue] (Y) {$Y$};
    \node[left=.5cm of X] (nil) {};
    \node[right=.5cm of Y] (lin) {};
    \draw[->] (nil) to (X);
    \draw[->, bend right] (X) edge node[below] {$b$} (Y);
    \draw[->, bend right] (Y) edge node[above] {$a$} (X);
    \draw (X) edge [loop above] node {$a$} (X);
\draw (Y) edge [loop above] node {$b$} (Y);
    \end{tikzpicture}
\\
\noalign{\smallskip}
         \orange \mu X (aX \blue + bX \blue + \mgt \nu Y (bY)) &
          \mgt \nu X \orange \mu Y (aX \blue + bY) 
    \end{array}
    \]
    \todo{spacing hacky below, better to fix}
    \[
\begin{tikzpicture}[baseline=(root)]
    \node[circle, draw, color=red] (root) at (0,0) {$\phantom{\cap\  }$ };
    \node[right=.5cm of root, yshift=.9cm, circle, draw, color=blue, label=below:{$\mgt 0$}] (Xa) {$X_a$};
    \node[right=1cm of Xa, circle, draw, color=blue,label=below:{$\orange 1$}] (Ya) {$Y_a$};
        \node[left=.5cm of root] (nil) {};
    \node[right=.5cm of Ya] (lin) {};
    \draw[->] (nil) to (root);
    
    \node[right=.5cm of root, yshift=-.9cm, circle, draw, color=blue,label=above:{$\orange 1$}] (Xb) {$X_b$};
    \node[right=1cm of Xb, circle, draw,label=above:{$\mgt 2$}] (Yb) {$Y_b$};
    
    \draw[->] (root) edge[bend left] (Xa);
    \draw[->] (Xa) edge[bend left] node[above] {$b$} (Ya);
    \draw[->] (Ya) edge[bend left] node[below] {$a$} (Xa);
    \draw (Xa) edge[loop above] node[above] {$a$} (Xa);
    \draw (Ya) edge[loop above] node[above] {$b$} (Ya);

    \draw[->] (root) edge[bend right] (Xb);
    \draw[->] (Xb) edge (Yb);
    \draw (Xb) edge[loop below] node[below] {$a,b$} (Xb);
    \draw (Yb) edge[loop below] node[below] {$a$} (Yb);


    \node[left=.7cm of root, yshift=1cm] {\emph{infinitely many $a$s}};
    \node[left=.7cm of root] {$\text{\emph{and}} \ \ \ \ \ \ \ \ \  $ :};
    \node[left=.7cm of root, yshift=-1cm] {\emph{finitely many $b$s}$\ \ $};

    \node[right=.5cm of Ya] {$\ \ \ \mgt \nu X_a \orange \mu Y_a \left(aX_a \blue + bY_a\right)$};
    \node[right=.5cm of Ya, yshift=-.9cm] {$\ \ \ \ \ \ \ \ \ \ \ \ \ \  \red \cap $};
    \node[right=.5cm of Yb] {$\orange \mu X_b (aX_b \blue + bX_b \blue + \mgt \nu Y_b (bY_b)$};
\end{tikzpicture}
\]

\end{example}

The preceding example suggested an association between expressions and APAs. Let us now make this more formal.

\begin{definition}
[Fischer-Ladner]
\label{def:fl}
Define $\flred$ as the smallest relation on RLL expressions satisfying:
\begin{itemize}
    \item $ae \flred e$.
    \item $e_0\star e_1 \flred e_i$, for $i\in \{0,1\}$ and $\star \in \{+,\cap\}$.
    \item $\sigma X e(X) \flred e(\sigma Xe(X))$, for $\sigma \in \{\mu,\nu\}$.
\end{itemize}
Write $\leqfl$ for the reflexive transitive closure of $\flred$.
The \defname{Fischer-Ladner ($\FL$) closure} of an expression $e$, written $\fl e$, is $\{f\leqfl e\}$.
We also write $e\subform f$ if $e$ is a subformula of $f$, in the usual sense.
\end{definition}

It is well-known that $\fl e $ is always finite. 
This follows by induction on the structure of $e$, relying on the equality $\fl{\sigma X e} = \{\sigma X e\} \cup \{f [ \sigma X e / X] : f \in \fl e\}$ (see, e.g., \cite{DD24a,DasDe24:rla-preprint} for further details).

From here we can readily associate to any expression $e$ an automaton $\mathbf A_e$ with:
\begin{itemize}
    \item \emph{States}:  $\fl e$, with expressions $0,f+g$ existential and expressions $\top,f\cap g$ universal.\footnote{Again, it does not matter whether other expressions are existential or universal states, as there is a unique instance of $\flred$ from them.} The initial state is $e$.
    \item \emph{Transitions}: 
    \begin{itemize}
        \item $af \underset a \to f$ whenever $af \in \fl e$; and,
        \item $g \to g'$ whenever $g\flred g'$ and $g$ is not of form $af$.
    \end{itemize} 
    \item \emph{Colouring}: any function $c_e: \fl e \to \Nat $ s.t.:
    \begin{itemize}
        \item $c_e$ is monotone wrt subformulas, i.e.\ if $f\subform g \implies c(f) \leq c(g)$; and,
        \item $c_e$ assigns $\mu$ and $\nu$ formulas odd and even numbers, respectively, i.e.\ always $c_e(\mu Xf(X))$ is odd and $c_e(\nu X f(X))$ is even.
    \end{itemize}
\end{itemize}



\todo{commented dependency colouring here}

The APAs given in \cref{ex:apas} are just simplifications of $\mathbf A_e$, for the given associated expression $e$.
As expected we have:

\begin{theorem}
\label{thm:apa-from-expr-is-equiv}
    For closed RLL expressions $e$, $\lang e = \lang{\mathbf A_e}$.
\end{theorem}
We shall delay the justification of this result to the next section, after building up some game theoretic machinery for RLL expressions mirroring the well-known \emph{acceptance games} for APAs. 
Before that let us point out that, as a consequence of this result together with \cref{thm:apa-omega-regular}, we duly obtain the converse of \cref{thm:apa-omega-regular}:

\begin{corollary}
    For any closed RLL expression $e$, $\lang e$ is $\omega$-regular.
\end{corollary}

\begin{remark}
    [Automata to expressions]
    In fact \cref{prop:omega-reg-have-munu-exps} can itself be refined: RLL expressions can really be construed as a \emph{notation} for APAs. 
The representation of automata by right-linear expressions is detailed in \cite{DD24a,DasDe24:rla-preprint} for the case of non-deterministic finite word automata (NFAs) by RLA expressions. Essentially same argument works to represent APAs by RLL expressions, but a detailed development is beyond the scope of this work. 
At a very high level, an APA can be construed as a `hierarchical' system of equations (one for each transition), with states construed as variables (such systems have appeared in, e.g., \cite{BruseFL15:guarded-transformations,SeiNeu99:guard-nest-fps}). This system can be solved by (closed) fixed point expressions using \emph{Beki\'c's Lemma} \cite{Bekic84}. In the case of NFAs, the precise order in which we solve the variables is unimportant, as RLA expressions have only $\mu$ fixed points, not $\nu$. 
For APAs the order is forced by the colouring, requiring us to solve lower coloured states first, by $\nu$ if the state is coloured even and by $\mu$ if the state is coloured odd. \todo{could develop this formally in appendices}
\end{remark}

\section{Evaluation games}
\label{sec:eval-games}

As an engine for the necessary metalogical results later it is useful to first develop \emph{game-theoretic} characterisations of formula evaluation.
As a consequence we also recover a proof of \cref{thm:apa-from-expr-is-equiv}.

\subsection{More on Fischer-Ladner}

Write $\flredeq$ for the reflexive closure of $\flred$, i.e.\ $e \flred f $ if $e=f $ or $e\flred f$.
A \defname{trace} is a sequence $e_0 \flredeq e_1 \flredeq \cdots$.
We also write $e\lefl f$ if $e \leqfl f \not \leqfl e$.

We mentioned some properties of the Fischer-Ladner closure in the previous section. Let us collect these and more into a formal result:

\begin{proposition}
[Properties of $\FL$, see, e.g., \cite{StrEme89:aut-th-proc-mu-calc,KupMarVen22:fl-props}]
\label{prop:fl-props}
    We have:
    \begin{enumerate}
        \item\label{item:fl-finite} $\fl e$ is finite, and in fact has size linear in that of $e$.
        \item\label{item:fl-preorder} $\leqfl$ is a preorder and $\lefl$ is well-founded.
        \item\label{item:traces-have-least-inf-occ-elem} Every trace has a minimum infinitely occurring element, under $\subform$. 
        If a trace is not eventually stable, the minimum element has form $\mu X e $ or $\nu X e$.
    \end{enumerate}
\end{proposition}

\begin{proof}
    [Proof idea]
    \ref{item:fl-finite} follows by straightforward structural induction on $e$, noting that $\fl{\sigma X e}= \{\sigma X e\} \cup \{f[\sigma X e /X] : f \in\fl{e} \}$.
    \ref{item:fl-preorder} is immediate from the definitions.  
    For \ref{item:traces-have-least-inf-occ-elem} note that $\flredeq \ \subseteq\  \subform \cup \supform$, whence the property reduces to a more general property on well partial orders: any path along $\subform \cup \supform$ must have a $\subform$-minimum.
\end{proof}

We call the smallest infinitely occurring element of a trace its \defname{critical} formula.
If a trace is not ultimately stable, we call it a \defname{$\mu$-trace} or \defname{$\nu$-trace} if its critical formula is a $\mu$-formula or a $\nu$-formula, respectively.

\subsection{The evaluation game}
In this subsection we define games for evaluating expressions, similar in spirit to \emph{acceptance games} for APAs. 

\begin{definition}
    [Evaluation Game]
    The \defname{Evaluation Game} is a two-player game, played by Eloise ($\Eloise$) and Abelard ($\Abelard$).
    The positions of the game are pairs $(w,e)$ where $w\in \Alphabet^\omega$ and $e$ is an expression. 
    The moves of the game are given in \cref{fig:eval-rules}.\footnote{For positions where a player is not assigned, the choice does not matter as there is a unique available move.}

    An infinite play of the evaluation game is \defname{won} by $\Eloise$ (aka \defname{lost} by $\Abelard$) if the smallest expression occurring infinitely often (in the right component) is a $\nu$-formula.
    (Otherwise it is won by $\Abelard$, aka lost by $\Eloise$.)

    If a play reaches deadlock, i.e.\ there is no available move, then the player who owns the current position loses.
\end{definition}

Note that property \eqref{item:traces-have-least-inf-occ-elem} from \cref{prop:fl-props} justifies our formulation of the winning condition in the evaluation game: the right components of any play always form a trace that is never stable, by inspection of the available moves. 
Thus it is either a $\mu$-trace or a $\nu$-trace.

Note that winning can be formulated as a parity condition, assigning priorities consistent with the subformula ordering and with $\mu$ and $\nu$ formulas having odd and even priorities, respectively, just like for the APAs $\mathbf A_e$ we defined earlier.
It is well-known that parity games are positionally determined,  i.e.\ if a player has a winning strategy from some position, then they have one that depends only on the current position, not the previous history of the play (see, e.g., \cite{GTW03,PerPin:inf-word-aut-book}). Thus:
\begin{observation}
\label{obs:eval-game-is-pos-det}
    The Evaluation Game is positionally determined.
\end{observation}
Indeed, by a standard well-ordering argument, there is a \emph{universal} positional winning strategy for $\Eloise$, one that wins from each winning position. Similarly for $\Abelard$.

\begin{figure}
    \centering
    \begin{tabular}{|c|c|c|}
        \hline
        Position  & Player &  Available moves
        \\\hline
       $(aw,ae)$ & - &  $(w,e)$ \\
       $(aw,be)$ with $a\neq b$ & $\Eloise$ &  \\
       $(w,0)$ & $\Eloise$ & \\
       $(w,\top)$ & $\Abelard$ & \\
       $(w,e+f)$ & $\Eloise$ & $(w,e)$, $(w,f)$ \\
       $(w,e\cap f)$ & $\Abelard$ & $(w,e)$, $(w,f)$ \\
       $(w, \mu X e(X))$ & -  & $(w,e(\mu X e(X))$ \\
       $(w, \nu X e(X))$ & - & $(w,e(\nu X e(X))$ 
       \\
        \hline
    \end{tabular}    
    \caption{Rules of the evaluation game.}
    \label{fig:eval-rules}
\end{figure}

As suggested by its name, the Evaluation Game is adequate for $\Lang$, the main result of this subsection:

\begin{lemma}
[Evaluation]
\label{lem:eval}
    $w \in \lang e $ $\iff$ Eloise  has a winning strategy from $(w,e)$. (Otherwise, by determinacy, Abelard has a winning strategy from $(w,e)$).
\end{lemma}

The proof of this result uses relatively standard but involved techniques, requiring a detour through a theory of approximants and signatures when working with fixed point logics, inspired by previous work on the modal $\mu$-calculus such as \cite{StrEme89:aut-th-proc-mu-calc,NiwWal96:games-mu-calc}. 
Roughly, for the $\implies$ direction, we construct a winning $\Eloise$-strategy by preserving language membership whenever making a choice at a $+$-state $(w,e+f)$.
However this is not yet enough: if \emph{both} $w \in \lang e$ and $w \in \lang f$, we must make sure to `decrease the witness' of membership. 
E.g.\ the $\Eloise$ strategy that loops on $(w, \mu X(\top+X)) $ does not win despite $w \in \lang {\mu X(\top+X)} = \lang \top = \Alphabet^\omega$: at some point we must choose the move $(w, \top + \mu X(\top+X)) \rightarrow (w, \top)$ to win.
Formally such a `witness' is given by an \emph{approximant} of a fixed point. For instance if $w \in \lang{\mu X e(X)}$ then we consider the least ordinal $\alpha$ such that $w \in \lang{e^\alpha(0)}$, appropriately defined. 
We can assign such approximations to \emph{every} least fixed point of an expression, \emph{signatures}, lexicographically ordered according to a `dependency order' induced by $\leqfl$, and always make choices at $+$-states according to least signatures.
The $\impliedby$ direction is completely dual, constructing a winning $\Abelard$-strategy, under determinacy, by approximating greatest fixed points instead of least.

We shall give a proof of \cref{lem:eval} in the next subsection, but the reader familiar with such results may safely skip it.
Before that, let us point out one useful consequence of the Evaluation Lemma: it yields immediately the $\omega$-regularity of languages denoted by RLL expressions:

\begin{proof}
    [Proof sketch of \cref{thm:apa-from-expr-is-equiv}]
    The evaluation game for an expression $e$ is just the acceptance game (see, e.g., \cite{Bojan23course}) for the APA $\mathbf A_e$.
    More directly, an $\Eloise$ strategy from $(w,e)$ is just a run-tree from $(w,e)$ in $\mathbf A_e$, and the former is winning if and only if the latter is accepting.
    From here we conclude by \cref{lem:eval}.
\end{proof}

\todo{MSOL proof commented below}

\subsection{Proof of the Evaluation Lemma}
A key point for proving \cref{lem:eval} is the fact that least and greatest fixed points admit a dual characterisation as limits of approximants.
%
    The Knaster-Tarski theorem tells us that, for any complete lattice $(L,\leq)$ and monotone operation $f:L\to L$, there is a least fixed point $\mu f = \bigwedge \{A\geq f(A)\}$ and a greatest fixed point $\nu f = \bigvee \{A \leq f(A)\}$. (More generally, the set $F$ of fixed points of $L$ itself forms a complete sublattice.)
    However $\mu f$ and $\nu f $ can alternatively defined in a more iterative fashion.

    First, for $A\in L$ and $\alpha $ an ordinal, define the \defname{approximants} $f^\alpha(A)$ and $f_\alpha(A)$ by transfinite induction on $\alpha$ as follows,
    \[
    \begin{array}{r@{\ \df \ }l}
         f^0(A) & A  \\
         f^{\alpha + 1}(A) & f(f^\alpha(A)) \\
         f^\lambda (A) & \bigvee\limits_{\alpha <\lambda}f^\alpha(A)
    \end{array}
\qquad
    \begin{array}{r@{\ \df \ }l}
         f_0(A) & A  \\
         f_{\alpha + 1}(A) & f(f_\alpha(A)) \\
         f_\lambda (A) & \bigwedge\limits_{\alpha <\lambda}f_\alpha(A)
    \end{array}
    \]
    where $\lambda$ ranges over limit ordinals. It turns out that we have
    \[
    \begin{array}{r@{\ = \ }l}
         \mu f & \bigvee\limits_{\alpha} f^\alpha(\bot_L) \\
        \nu f & \bigwedge \limits_{\alpha} f_\alpha (\top_L)
    \end{array}
    \]
    where $\bot_L$ and $\top_L$ are the least and greatest elements, respectively, of $(L,\leq)$, and $\alpha $ ranges over all ordinals. (In fact it suffices to bound the range by the cardinality of $L$, by the transfinite pigeonhole principle).

    This viewpoint often provides a more intuitive way to compute fixed points, in particular for calculating $\lang e$.
    For instance the calculations from \cref{ex:(in)finitely-many} are naturally recast as the computation of fixed points by approximants.\todo{say more? Give example again?}

\medskip

Now let us turn to proving \cref{lem:eval}. Recall the subformula ordering $\subform$ and the FL ordering $\leqfl$ we introduced earlier.
Let us introduce a standard ordering of fixed point formulas (see, e.g., \cite{StrEme89:aut-th-proc-mu-calc,KupMarVen22:fl-props}):

\begin{definition}
    [Dependency order] \label{def:dependency-order}
    The \defname{dependency order} on closed expressions, written $\dleq$, is defined as the lexicographical product $\leqfl \times \supform$.
    I.e.\ $e \preceq f$ if either $e\lefl f$ or $e\eqfl f $ and $f \subform e$.
\end{definition}\anupam{It would be better if the dependency order was consistent with the subformula ordering, due to the way progressing traces / winning plays are defined (in terms of subformula), so that the induced parity condition matches. This amounts to inverting the current formulation of dependency order: smaller expressions, according to the order, should be more important, not larger ones.}

Note that, by properties \ref{item:fl-finite} and \ref{item:fl-preorder} of \cref{prop:fl-props}, we have that $\dleq$ is a well partial order on expressions. 
In the sequel we assume an arbitrary extension of $\dleq$ to a total well-order $\leq$.

\begin{definition}
    [Signatures]
    Let $M$ be a finite set of $\mu$-formulas $\{\mu X_0 e_0  > \cdots > \mu X_{n-1} e_{n-1}\}$.
    An $M$-\defname{signature} (or $M$-\defname{assignment}) is a sequence $\vec \alpha$ of ordinals indexed by $M$.
    Signatures are ordered by the lexicographical product order.
    An $M$-\defname{signed} formula is an expression $e^{\vec \alpha}$, where $e$ is an expression and $\vec \alpha$ is an $M$-signature.
%
For $N$ is a finite set of $\nu$-formulas we define $N$-signatures similarly and use the notation $e_{\vec \alpha}$ for $N$-signed formulas.
\end{definition}

We evaluate signed formulas in $\Lang$ just like usual formulas, adding the clauses,
\begin{itemize}
\item $\wlang {(\mu X_i e_i(X))^{\vec \alpha_i 0 \vec \alpha^i}} \df \emptyset$.
    \item $\wlang {(\mu X_i e_i(X))^{\vec \alpha_i (\alpha_i + 1 ) \vec \alpha^i}} \df \wlang { (e_i (\mu X_i e_i(X)))^{\vec \alpha_i \alpha_i  \vec \alpha^i}}$.
    \item $\wlang{(\mu X_i e_i(X))^{\vec \alpha_i \alpha_i \vec \alpha^i}} \df \bigcup \limits_{\beta_i <\alpha_i} \wlang{(\mu X_i e_i(X))^{\vec \alpha_i \beta_i \vec \alpha^i}} $, when $\alpha_i$ is a limit.

    \medskip
    
\item $\wlang {(\nu X_i e_i(X))_{\vec \alpha_i 0 \vec \alpha^i}} \df \Alphabet^{\leq \omega}$.
    \item $\wlang {(\nu X_i e_i(X))_{\vec \alpha_i (\alpha_i + 1 ) \vec \alpha^i}} \df \wlang { (e_i (\nu X_i e_i(X)))_{\vec \alpha_i \alpha_i  \vec \alpha^i}}$.
    \item $\wlang{(\nu X_i e_i(X))_{\vec \alpha_i \alpha_i \vec \alpha^i}} \df \bigcap \limits_{\beta_i <\alpha_i} \wlang{(\nu X_i e_i(X))_{\vec \alpha_i \beta_i \vec \alpha^i}} $, when $\alpha_i$ is a limit.
\end{itemize}
where we are writing $\vec \alpha_i \df (\alpha_j)_{j<i}$ and $\vec \alpha^i \df (\alpha_j)_{j>i}$.

Since least and greatest fixed points can be computed as limits of approximants, and since expressions compute monotone operations in $\Lang$,
we have that, for any sets $M,N$ of $\mu,\nu$ formulas respectively:
\begin{itemize}
    \item $\lang e = \bigcup \limits_{\vec \alpha} \lang{e^{\vec \alpha}}$
    \item $\lang e = \bigcap \limits_{\vec \beta} \lang{e_{\vec \beta}}$
\end{itemize}
where $\vec \alpha $ and $\vec \beta$ range over all $M$-signatures and $N$-signatures, respectively.
Thus we have: \todo{say more by way of justification?}
\begin{proposition}
Suppose $e$ is an expression and $M,N$ the sets of $\mu,\nu$-formulas, respectively, in $\fl e$. We have:
\begin{itemize}
    \item If $w \in \wlang e$ then there is a least $M$-signature $\vec \alpha$ such that $w \in \wlang {e^{\vec \alpha}}$.
    \item If $w\notin \wlang e$ then there is a least $N$-signature $\vec \alpha$ such that $w \notin \wlang {e_{\vec \alpha}}$.
\end{itemize}
\end{proposition}

In fact, for RLL expressions interpreted in $\Lang$, it suffices to take only signatures of finite ordinals, i.e.\ natural numbers, for the result above, but we shall not use this fact.
We are now ready to prove our characterisation of evaluation:

\begin{proof}
[Proof sketch of \cref{lem:eval}]
Let $M,N$ be the sets of $\mu,\nu$-formulas, respectively, in $\fl e$. 

    $\implies$. 
    Suppose $w \in \wlang e $. 
    We construct a winning $\Eloise$ strategy $\strat e $ from $(w,e)$ by always preserving membership of the word in the language of the expression.
    Moreover, at each position $(w',e_0 + e_1)$, $\strat e $ chooses a summand $e_i$ admitting the least $M$-signature $\vec \alpha$ for which $w'\in \wlang {e_i^{\vec \alpha}}$.
    As $\strat e$ preserves word membership, no play reaches a state $(aw,be)$, with $a\neq b$, or $(w,0)$, and so any maximal finite play of $\strat e$ is won by $\Eloise$.
     So let $(w_i,e_i)_{i<\omega}$ be an infinite play of $\strat e$ and, for contradiction, assume that its smallest infinitely occurring formula is $\mu X f(X)$. 
    Write $\vec \alpha_i$ for the least $M$-signature s.t.\ $w_i\in \wlang {e_i^{\vec \alpha_i}}$, for all $i<\omega$.
    By construction $(\vec \alpha_i)_{i<\omega}$ is a monotone non-increasing sequence.
    Moreover, since $(e_i)_{i<\omega}$ is infinitely often $\mu Xf(X)$,
    the sequence $(\vec \alpha_i)_{i<\omega}$ does not converge. Contradiction.\todo{ more justification could be given here.}

    $\impliedby$. The argument is entirely dual, constructing an $\Abelard$-strategy $\strat a $ that preserves non-membership, following least $N$-signatures at positions $(w',e_0 \cap e_1)$. 
\end{proof}


 \section{A cyclic system for RLL expressions}
\label{sec:crll}

In this section we shall present a sequent style system and associated notions of `non-wellfounded' and `cyclic' proof, before proving soundness of cyclic proofs for $\Lang$.

\subsection{Some properties of the language model}
\label{sec:props-of-lang}
Let us take a moment to remark upon some principles valid in the intended interpretation $\Lang$ of RLL expressions, in order to motivate the proof system we are about to introduce. 
In what follows we construe $\Lang $ as a bona fide structure (in the model theoretic sense) with domain $\Pow(\Alphabet^\omega)$.
As usual we write $e\leq f \df e+f =f$, equivalently $e = e\cap f$ (so in $\Lang$, $\leq$ just means $\subseteq)$.

\begin{itemize}
    \item $(0,\top, +,\cap)$ forms a bounded distributive lattice:\footnote{Some of these axioms are redundant, but we include them all to facilitate the exposition.} 
    \begin{equation}
        \label{eq:dist-lattice}
        \begin{array}{r@{\ = \ }l}
        e + 0 & e \\
        e + (f + g) & (e + f ) + g\\
        e + f & f + e \\
        e+e & e \\
        e + (e\cap f) & e\\
        e + (f\cap g) & (e+ f)\cap(e+ g)
    \end{array}
    \qquad
     \begin{array}{r@{\ = \ }l}
         e \cap \top & e\\
         e \cap (f \cap g) & (e \cap f ) \cap g\\
         e \cap f & f \cap e\\
         e \cap e & e \\
         e \cap (e + f) & e \\
         e \cap (f+ g) & (e\cap f)+(e\cap g)
    \end{array}
    \end{equation}
\end{itemize}


\begin{itemize}
    \item Each $a\in \Alphabet$ is a (lower) semibounded lattice homomorphism:
    \begin{equation}
        \label{eq:letter-lbdd-lattice-homo}
         \begin{array}{r@{\ = \ }l}
     a0 & 0 \\
     a(e+f) & ae + af\\
    a(e\cap f) & ae \cap af
    \end{array}
    \end{equation}
\end{itemize}

In particular, of course $\Lang \not \models a\top = \top$, so in this sense $0$ and $\top$ do not behave dually in $\Lang$. 
Instead we have a variant of this principle, indicating that the homomorphisms freely factor the structure:

\begin{itemize}
    \item The ranges of $a\in\Alphabet$ partition the domain:
    \begin{equation}
    \label{eq:letters-freely-cogenerate-structure}
        \begin{array}{r@{\ = \ }ll}
         ae \cap bf & 0 & \text{whenver $a\neq b$} \\
         \top & \sum\limits_{a\in \Alphabet} a\top
    \end{array}
    \end{equation}
\end{itemize}

\todo{commented $\mu$,$\nu$ being least/greatest pre/post fixed points.}


\begin{figure}
    \textbf{$\Alphabet$ rules:}
    \[
    \vlinf{\dis}{a\neq b}{ae,bf \seqar }{}
    \qquad
    \vlinf{\kk a}{\Gamma \neq \emptyset}{a\Gamma \seqar a\Delta}{\Gamma\seqar \Delta}
    \qquad
    \vlinf{\exh}{}{\seqar \{a\Gamma_a\}_{a\in \Alphabet}}{\{\seqar \Gamma_a\}_{a\in \Alphabet}}
    \]
    
    \medskip
    \textbf{Structural rules:}
    \[
    \vlinf{\lr\wk}{}{\Gamma,e\seqar \Delta}{\Gamma \seqar \Delta}
    \qquad
    \vlinf{\rr\wk}{}{\Gamma\seqar \Delta,e}{\Gamma \seqar \Delta}
    \]

    \medskip
    \textbf{Left logical rules:}
    \[
    \begin{array}{ccc}
         \vlinf{\lr 0}{}{\Gamma,0 \seqar \Delta}{}
         &
         \vliinf{\lr +}{}{\Gamma,e + f \seqar \Delta}{\Gamma,e\seqar \Delta}{f \seqar \Delta}
         &
         \vlinf{\lr \mu }{}{\Gamma,\mu X e(X) \seqar \Delta}{\Gamma,e(\mu X e(X)) \seqar \Delta}
         \\
         \noalign{\medskip}
          \vlinf{\lr \top}{}{\Gamma, \top\seqar \Delta}{\Gamma \seqar \Delta}
          &
          \vlinf{\lr \cap}{}{\Gamma,e \cap f \seqar \Delta}{\Gamma,e,f\seqar \Delta}
          &
          \vlinf{\lr \nu }{}{\Gamma,\nu X e(X) \seqar \Delta}{\Gamma,e(\nu X e(X)) \seqar \Delta}
    \end{array}
    \]
    
    \medskip
    \textbf{Right logical rules:}
    \[
    \begin{array}{ccc}
    \vlinf{\rr 0}{}{\Gamma \seqar \Delta, 0}{\Gamma \seqar \Delta}
    &
    \vlinf{\rr +}{}{\Gamma\seqar \Delta, e+f}{\Gamma\seqar \Delta, e,f}
    &
    \vlinf{\rr \mu}{}{\Gamma\seqar \Delta, \mu X e(X)}{\Gamma\seqar \Delta, e (\mu X e(X))}
    \\
    \noalign{\medskip}
    \vlinf{\rr\top}{}{\Gamma\seqar \Delta,\top}{}
    &
     \vliinf{\rr\cap}{}{\Gamma\seqar\Delta,e\cap f}{\Gamma\seqar\Delta,e}{\Gamma\seqar\Delta,f}
     &
     \vlinf{\rr \nu}{}{\Gamma\seqar \Delta, \nu X e(X)}{\Gamma\seqar \Delta, e (\nu X e(X))}
    \end{array}
    \]
    \caption{Rules of the system $\LRLLLanghat$.}
    \label{fig:lrllhat}
\end{figure}

We shall now use these principles to design a proof system for comparing RLL expressions over $\Lang$.

\subsection{A sequent system}
A \defname{sequent} is an expression $\Gamma\seqar \Delta$, where $\Gamma$ and $\Delta$ are sets of expressions (called \defname{cedents}). The intended reading of the LHSs and RHSs of sequents are the meets and sums of their elements, respectively, i.e.\ a sequent $\Gamma \seqar \Delta$ is associated with the inequality $\bigcap \Gamma \leq \sum \Delta$.

\begin{definition}
   The system $\LRLLLanghat$ is given by the rules of \cref{fig:lrllhat}, where we write $a\Gamma \df \{ae : e\in \Gamma\}$.
As usual for structural and logical steps \defname{principal} formula is the distinguished formula on the LHS or RHS, respectively, of the lower sequent, as typeset in \cref{fig:lrllhat}.
Any distinguished formulas on the LHS or RHS, respectively, of the upper sequent are called \defname{auxiliary}.
\end{definition}

Let us take a moment to justify some of the rules of $\LRLLLanghat$:
\begin{itemize}
    \item The $\Alphabet$-rules are justified by the homomorphism principles in $\Lang$, \cref{eq:letter-lbdd-lattice-homo,eq:letters-freely-cogenerate-structure}.
    In particular $\dis$ and $\exh$ are justified by \eqref{eq:letters-freely-cogenerate-structure}, while the $\K$ rules are justified by \eqref{eq:letter-lbdd-lattice-homo}.
 Note the side condition on $\kk a $ that the LHS is nonempty: this corresponds to $a$ being a lower-bounded lattice homomorphism, but not one that preserves $\top$ (indeed by $\dis$ when $|\Alphabet|\geq 1$).
 \item 
The left and right logical rules are justified by the bounded distributive lattice principles in $\Lang$ cf.~\cref{eq:dist-lattice}. In particular distributivity is required as the LHS and RHS contexts may be nonempty.
\item Finally the $\mu$ and $\nu $ rules are justified by the fact that $\mu$ and $\nu$ compute fixed points in $\Lang$, cf.~\cref{dfn:lang-semantics} and the discussion thereafter.
\end{itemize}

Putting these together we obtain:
\begin{proposition}
    [Local soundness]
    \label{prop:local-soundness}
    Each rule of $\LRLLLanghat$ is sound for $\Lang$. I.e.\ for each inference step 
    $$\vliiinf{\mathsf r }{}{\Gamma \seqar \Delta}{\Gamma_1\seqar \Delta_1}{\cdots}{\Gamma_{n} \seqar \Delta_{n}}$$ 
    of $\LRLLLanghat$, if $\lang {\bigcap \Gamma_i} \subseteq \lang {\sum \Delta_i}$ for $i=1,\dots,n$ then $ \lang{\bigcap \Gamma} \subseteq \lang{\sum \Delta}$.
\end{proposition}

However notice that we have not included any induction or coinduction rules in $\LRLLLanghat$: in fact this system does not distinguish $\mu X e $ and $\nu X e$ at the level of rules.\footnote{The `hat' notation is common in the literature on theories of inductive definitions for systems of fixed points that are not necessarily extremal. See, e.g., \cite{iid-book}.}
Rather their distinction is controlled by a notion of infinite proof that we now turn to.

\subsection{Non-wellfounded and cyclic proofs}

We introduce a notion of non-wellfounded proof that allows us to recover (co)inductive reasoning:
\begin{definition}
[Preproofs]
A \defname{preproof} (of $\LRLLLanghat$) is generated \emph{coinductively} from the rules of $\LRLLLanghat$. 
I.e.\ it is a possibly infinite tree of sequents generated by the rules of $\LRLLLanghat$.
A preproof is \defname{cyclic} or \defname{regular} if it has only finitely many distinct sub-preproofs.
\end{definition}

While usual inductively generated proofs can be represented as finite trees or dags of sequents, cyclic preproofs are rather represented as \emph{graphs} of sequents, possibly with cycles.
Let us see some basic examples now, before turning to more interesting ones later.

\begin{example}
    [Empty and universal, revisited]
    \label{ex:0-top-entailments}
    Recall the expressions $\mu X X $ and $\nu X X $ from \cref{ex:empty-and-universal}, respectively denoting the empty and universal languages in $\Lang$.
    Here are some preproofs involving them,
    \[
    \vlderivation{
    \vlin{\infrule}{\bullet}{\mu XX \seqar \nu XX}{
    \vlin{\infrule}{\bullet}{\mu XX \seqar \nu XX}{\vlhy{\vdots}}
    }
    }
    \qquad
    \vlderivation{
    \vlin{\infrule'}{\bullet}{\nu XX \seqar \mu XX}{
    \vlin{\infrule'}{\bullet}{\nu XX \seqar \mu XX}{\vlhy{\vdots}}
    }
    }
    \qquad
    \vlderivation{
    \vlin{\pi_0}{}{\mu XX \seqar \nu XX}{
    \vlin{\pi_1}{}{\mu XX \seqar \nu XX}{\vlhy{\vdots}}
    }
    }
    \]
    where $\infrule $ is either $\lr \mu$ or $\rr \nu$, and $\infrule'$ is either $\lr \nu$ or $\rr \mu$, and $\pi_i$ is $\lr \mu$ if the $i$\textsuperscript{th} bit of the binary expansion of $\pi$ is 0, otherwise $\pi_i$ is $\rr\nu$.
    Here we have used $\bullet$ to indicate roots of identical subpreproofs, and we shall use similar notation for regular preproofs throughout.
    The first two preproofs above are (always) regular, in particular with each subproof being identical.
    The third preproof is not regular, as $\pi$ is irrational.
\end{example}

Notice that $\LRLLLanghat$ does not include a native identity rule $\vlinf{\id}{}{e \seqar e}{}$. For closed expressions $e$, such identities may be derived by (infinite) preproofs in $\LRLLLanghat$.
To see this we first establish a stronger property: 
\begin{proposition}
[Functors]
\label{functors}
    For each expression $e(\vec X)$ (all free variables indicated), there are regular preproofs of $e(\vec f) \seqar e(\vec g)$ using additional initial sequents of form $f_i \seqar g_i$.
\end{proposition}
\begin{proof}
[Proof sketch]
    We proceed by induction on the structure of $e(\vec X)$. 
    All cases are routine except for the fixed point cases.
 If $e(\vec X) = \mu X e' (X,\vec X)$ we construct,
\[
\toks0={.4}
\vlderivation{
\vlin{\lr \mu, \rr \mu}{\bullet}{\mu X e'(X,\vec f) \seqar \mu X e'(X,\vec g)}{
\vltrf{\IH}{e'(\mu Xe'(X,\vec f), \vec f) \seqar e'(\mu Xe'(X, \vec g), \vec g)}{
    \vlin{}{\bullet}{\mu X e'(X,\vec f) \seqar \mu X e'(X,\vec g)}{\vlhy{\vdots}}
}{\vlhy{}}{\vlhy{\{f_i \seqar g_i\}_i}}{\the\toks0}
}
}
\]
where $\IH$ is obtained by the inductive hypothesis. 
The case for greatest fixed points is similar.
\end{proof}

By setting $\vec X$ to be empty in the result above, we duly obtain:
\begin{corollary}
[Identity]
\label{identity}
There are regular preproofs of $e\seqar e$, for closed expressions $e$.    
\end{corollary}

Of course non-wellfounded reasoning can be fallacious, so genuine `proofs' must satisfy further conditions.
This is given by a `global trace condition' as per usual in cyclic proof theory:

\begin{definition}[Ancestry and traces]
    For an inference step $\infrule$, a formula $e'$ in the LHS/RHS of a premiss is an \defname{immediate ancestor} of a formula $e$ in the LHS/RHS, respectively, of the conclusion if:
\begin{itemize}
    \item $\infrule$ is structural or logical, $e$ is principal and $e'$ is auxiliary; or,
    \item $\infrule$ is a $\kk a$ step and $e = ae'$; or,
    \item $\infrule $ is a $\exh$ step, $e'$ occurs in the $a$-premiss and $e = ae'$; or,
    \item $e = e'$.
\end{itemize}

    For a branch $B$ (along some preproof), a \defname{$B$-trace} (or just `trace' when unambiguous) is a maximal path along the graph of immediate ancestry, restricted to $B$.
    An LHS or RHS trace is one that remains in the LHS or RHS respectively.
    We say that an infinite trace $\tau$ is \defname{progressing} if it is infinitely often principal and:
    \begin{itemize}
        \item $\tau$ is LHS with smallest infinitely often principal formula a $\mu$ formula; or,
        \item $\tau $ is RHS with smallest infinitely often principal formula a $\nu$ formula.
    \end{itemize}
\end{definition}

Note that, by construction, any $B$-trace as defined above is indeed a trace in the sense of \cref{def:fl}.
Thus, by \cref{prop:fl-props}, the notion of progress above is well-defined.\footnote{Note that the corner case when a principal formula has an identical auxiliary formula, like in \cref{ex:0-top-entailments}, is possible only when the formula in question is already a $\mu$ or $\nu$ formula.}
Let us also point out that:
\begin{observation}
\label{obs:guarded-preproofs-no-stable-traces}
If $\Gamma , \Delta$ consists of only guarded formulas and $P$ is a $\LRLLLanghat$ preproof of $\Gamma\seqar \Delta$, then no trace in $P$ is eventually stable.
\end{observation}
This means that any trace in a preproof of a guarded sequent is either a $\mu$-trace or a $\nu$-trace, cf.~\cref{prop:fl-props}.

\begin{definition}[$\CRLLLang$ system]
 A $\LRLLLanghat$ preproof is \defname{progressing} if each infinite branch has a progressing trace. Write $\CRLLLang$ for the class of regular progressing $\LRLLLanghat$-preproofs, which we may simply call $\CRLLLang$-proofs henceforth.
 We may write, say, $\CRLLLang \proves \Gamma \seqar \Delta$ if there is a $\CRLLLang$-proof of the sequent $ \Gamma \seqar \Delta$.
\end{definition}

\begin{example}
    [Some (non-)proofs]
    Looking again at the preproofs from \cref{ex:0-top-entailments}, we have:
    \begin{itemize}
        \item The first preproof is progressing regardless of whether $\infrule$ is $\lr \mu$ or $\rr\nu$.
        In the former case the only infinite branch has a progressing LHS trace, whereas in the latter case it has a progressing RHS trace.
        \item The second preproof is not progressing, regardless of whether $\infrule'$ is $\lr\nu$ or $\rr\mu$.
        In the former case the only infinitely often principal trace along the only infinite branch is a LHS $\nu$-trace, and in the latter case it is a RHS $\mu$-trace.
        \item The third preproof is indeed progressing (despite not being regular). The binary expansion of $\pi$ must have infinitely many $0$s and infinitely many $1$s, as it is irrational, and so the only infinite branch has a LHS $\mu$-trace and a RHS $\nu$-trace, both of which are progressing.
    \end{itemize}
\end{example}

\begin{remark}
    [Proof checking]
    \label{proof-checking}
    While the progressing condition for proofhood may look complex, let us emphasise that it is indeed decidable for regular preproofs. 
    The set of progressing branches of a regular preproof $P$ forms an $\omega$-regular language over the alphabet of (the finitely many) sequents in $P$.
    In particular we can construct a non-deterministic B\"uchi automaton $\mathbf B_P$ that guesses a progressing trace on the fly (along with its critical LHS-$\mu$ or RHS-$\nu$ formula), verifying that no smaller formula is unfolded after some finite prefix (which is also guessed). 
    Now $P$ is progressing just if $\mathbf B_P$ accepts \emph{every} infinite branch of $P$; this reduces simply to checking universality of a non-deterministic B\"uchi automaton, and so thus constitutes a $\mathbf{PSPACE}$ algorithm for proof checking.
The reader may consult, e.g.\ \cite{NiwWal96:games-mu-calc,DekKloMarVen23:prf-sys-mu-calc-via-det}, for similar developments in the setting of the modal $\mu$-calculus.
\end{remark}

\begin{remark}
    [Identity and functors, revisited]
    The functor preproofs constructed in \cref{functors} are indeed progressing. This is readily established by adding as an invariant to the Proposition statement `such that, for every branch from the conclusion to an initial sequent $f_i\seqar g_i$, there is a (finite) LHS/RHS trace from the conclusion's LHS/RHS to $f_i$/$g_i$ along which $f_i$/$g_i$ are always subformulas, respectively'.
    Thus $\CRLLLang$ proves $e\seqar e$ for each closed expression $e$.
\end{remark}








\subsection{Further examples}
We conclude this section with some examples of cyclic proofs more interesting for $\omega$-language theory.
 Recall the expressions $f_a \df \mu X (aX + bX + \nu Y(bY))$ and $i_a \df \nu X \mu Y (aX + bY)$ from \cref{ex:(in)finitely-many}, denoting the languages of finitely many $a$s and infinitely many $a$s respectively, when $\Alphabet = \{a,b\}$.
 Let us point out that both $f_a$ and $i_a$ are guarded.
 In what follows we write $i_a'\df \mu Y (ai_a + bY)$.

 Recall $\lang{\nu X(aX)} = a^\omega$ from \cref{ex:omega-iteration}.
 Our first example is a cyclic proof demonstrating that $a^\omega$ indeed has infinitely many $a$s:
 \[
 \vlderivation{
 \vlin{\lr\nu, \rr\nu}{\bullet}{\nu X (aX) \seqar \tikzmarknode{i1}{i_a}}{
 \vlin{\rr \mu}{}{a \nu X (aX) \seqar \tikzmarknode{i2}{i_a'}}{
 \vlin{\rr+,\rr\wk}{}{a\nu X (aX) \seqar ai_a \tikzmarknode{i3}{+} bi_a'}{
 \vlin{\kk a }{}{a\nu X (aX) \seqar \tikzmarknode{i4}{ai_a}}{
 \vlin{\lr\nu,\rr\nu}{\bullet}{\nu X (aX) \seqar \tikzmarknode{i5}{i_a}}{\vlhy{\vdots}}
 }
 }
 }
 }
 }
 \thread[orange]{(i1) (i2) (i3) (i4) (i5)}
 \]
The only infinite branch, looping on $\bullet$, has a progressing RHS trace with critical formula $i_a$, according to the indicated orange path.
Note in particular that, while $i_a'$ is also infinitely often principal, it contains $i_a$ as a subformula.

On the other hand here is a cyclic proof that $a^\omega$ does not have finitely many $a$s, using $\cap$:
\[
\vlderivation{
\vlin{\lr\cap}{}{f_a\cap \nu Y(aY) \seqar}{
    \vlin{\lr\nu}{\bullet}{\tikzmarknode{f1}f_a,\nu Y(aY) \seqar}{
        \vlin{\lr\mu}{}{\tikzmarknode{f2}f_a,a\nu Y(aY) \seqar}{
            \vliiin{\lr+}{}{af_a+b\tikzmarknode{f3}f_a+\nu Y(bY),a\nu Y(aY) \seqar}{
                \vlin{\kk a}{}{a\tikzmarknode{f4}f_a,a\nu Y(aY) \seqar}{
                    \vlin{}{\bullet}{\tikzmarknode{f5}f_a,\nu Y(aY) \seqar}{\vlhy\vdots}
                }
            }{
                \vlin{\dis}{}{bf_a,a\nu Y(aY) \seqar}{\vlhy{}}
            }{
                \vlin{\lr\nu}{}{\nu YbY,a\nu Y(aY) \seqar}{
                    \vlin{\dis}{}{b\nu YbY,a\nu Y(aY) \seqar}{\vlhy{}}
                }
            }
        }
    }
}
}
\thread[blue]{(f1) (f2) (f3) (f4) (f5)}
\]
Once again there is only one infinite branch, looping on $\bullet$, and this has a progressing LHS trace with critical formula $f_a$, according to the indicated blue path.

The next proof demonstrates that if an infinite word over $\{a,b\}$ has finitely many $a$s then it must have infinitely many $b$s.
\[
\vlderivation{
\vlin{\rr\nu, \rr\mu}{\bullet}{f_a\seqar \tikzmarknode{s1}{i_b}}{
\vlin{\lr \mu}{\circ}{\tikzmarknode{t2}{f_a} \seqar bi_b \tikzmarknode{s2}{+} ai_b'}{
\vliiin{\lr +, \rr+}{}{\tikzmarknode{t3}{af_a} + bf_a + \nu Y(bY) \seqar bi_b \tikzmarknode{s3}{+} ai_b'}{
    \vlin{\rr \wk}{}{\tikzmarknode{t4}{af_a} \seqar bi_b , \tikzmarknode{u4}{ai_b'}}{
    \vlin{\kk a}{}{\tikzmarknode{t5}{af_a} \seqar \tikzmarknode{u5}{a i_b'}}{
    \vlin{\rr \nu,\rr\mu}{}{\tikzmarknode{t6}{f_a} \seqar \tikzmarknode{u6}{i_b'}}{
    \vlin{\lr \mu}{\circ}{\tikzmarknode{t7}{f_a} \seqar \tikzmarknode{u7}{bi_b + ai_b'}}{\vlhy{\vdots}}
    }
    }
    }
}{
    \vlin{\rr\wk}{}{bf_a \seqar \tikzmarknode{s4}{bi_b} , ai_b'}{
    \vlin{\kk b}{}{bf_a \seqar \tikzmarknode{s5}{bi_b}}{
    \vlin{\rr\nu, \rr\mu}{\bullet}{f_a \seqar \tikzmarknode{s6}{i_b}}{\vlhy{\vdots}}
    }
    }
}{
    \vlin{\lr \nu, \rr\wk}{}{\nu Y (bY) \seqar bi_b , ai_b'}{
    \vlin{\kk b}{}{b\nu Y(bY) \seqar bi_b}{
    \vliq{}{}{\nu Y (bY) \seqar i_b}{\vlhy{}}
    }
    }
}
}
}
}
\thread[orange]{(t2) (t3) (t4) (t5) (t6) (t7)}
\thread[blue]{(s1) (s2) (s3) (s4) (s5) (s6)}
\thread[green]{(s2) (s3) (u4) (u5) (u6) (u7)}
\]
The rightmost proof of $\nu Y(bY) \seqar i_b$ is given by the previous example.
Besides the unique infinite branch along that subproof, there are now continuum many infinite branches, indexed by elements of $\{\bullet,\circ\}^\omega$ in the natural way. 
To justify that the cyclic proof above is indeed progressing, let us do a case analysis on infinite branches:
\begin{itemize}
    \item A branch that hits $\bullet $ only finitely often, eventually looping on $\circ$, has a progressing LHS trace with critical formula $f_a$, according to the orange path (eventually).
    \item A branch that hits $\circ$ only finitely often, eventually looping on $\bullet$, has a progressing RHS trace with critical formula $i_b$, according to the blue path (eventually).
    \item Otherwise an infinite branch must repeat both $\bullet$ and $\circ$ infinitely often, in which case it has a progressing RHS trace along $i_b$, according to the green or blue path, depending on the current loop, $\circ$ or $\bullet $ respectively. 
    Note that this trace only `progresses' on the $\bullet $ loop, along which $i_b$ is principal.
\end{itemize}

In a similar vein, here is a cyclic proof showing that any infinite word over $\{a,b\}$ has either finitely many $a$s or infinitely many $a$s:
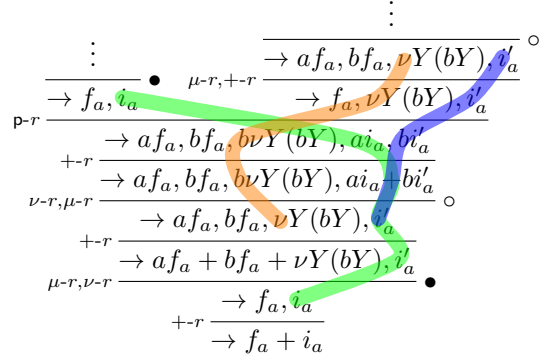
    \[
    \hspace{-2em}
\vlderivation{
    \vlin{\rr +}{}{ \seqar f_a+i_a}{
        \vlin{\rr\mu,\rr\nu}{\bullet}{\seqar f_a,\tikzmarknode{i1}i_a}{
            \vlin{\rr +}{}{\seqar af_a+bf_a+\nu Y (bY),\tikzmarknode{i2}i_a'}{
                \vlin{\rr\nu,\rr\mu}{\circ}{\seqar af_a,bf_a,\tikzmarknode{b1}\nu Y (bY),\tikzmarknode{i3}i_a'}{
                    \vlin{\rr +}{}{\seqar af_a,bf_a,\tikzmarknode{b2}b\nu Y (bY),ai_a\tikzmarknode{i4}+bi_a'}
                    {\vliin{\exh}{}{\seqar af_a,bf_a,\tikzmarknode{b3}b\nu Y (bY),\tikzmarknode{i5}ai_a,\tikzmarknode{j5}bi_a'}{
                        {
                            {
                                \vlin{}{\bullet}{\seqar f_a,\tikzmarknode{i7}i_a}{\vlhy{ \ \ \ \ \ \vdots\ \ \ \ \  }}
                            }
                        }
                    }{
                        {
                            {
                                \vlin{\rr\mu,\rr +}{}{\seqar f_a,\tikzmarknode{b5}\nu Y (bY),\tikzmarknode{j7}i_a'}{
                                    \vlin{}{\circ}{    \seqar af_a,bf_a,\tikzmarknode{b6}\nu Y (bY),\tikzmarknode{j8}i_a'}{\vlhy{\vdots}}
                                }
                            }
                        }
                    }}
                }
            }
        }
    }
}
\thread[green]{(i1) (i2) (i3) (i4) (i5) (i7)}
\thread[blue]{(i3) (i4) (j5) (j7) (j8)}
\thread[orange]{(b1) (b2) (b3) (b5) (b6)}
\]

Again there are continuum many infinite branches.
Let us again conduct a case analysis on such infinite branches:
\begin{itemize}
    \item A branch that hits $\bullet $ only finitely often, eventually looping on $\circ$, has a progressing trace with critical formula $\nu Y(bY)$, according to the orange path (eventually).
    \item A branch that hits $\circ$ only finitely often, eventually looping on $\bullet$, has a progressing trace with critical formula $i_a$, according to the green path (eventually).
    \item Otherwise an infinite branch must repeat both $\bullet$ and $\circ $ infinitely often, in which case it has a progressing trace with critical formula $i_a$, according to the green and blue paths depending on the current loop, $\bullet $ or $\circ$ respectively.
\end{itemize}

Finally let us consider the sequent $f_a\cap f_b \seqar$ (``it is not possible for an infinite word over $\{a,b\}$ to have finitely many $a$s and finitely many $b$s''). We give its cyclic proof in \cref{fig:fa-cap-fb-empty}, where every infinite branch actually has two progressing LHS traces, one with critical formula $f_a$ and one with critical forula $f_b$.\todo{could draw traces}

\begin{figure}
    \centering
    \scalebox{0.95}{
\rotatebox{90}{
$
\vlderivation{
  \vlin{\lr\cap}{}{f_a\cap f_b}{
    \vlin{\lr\mu}{\bullet}{f_a,f_b}{
      \vlin{\lr\mu}{}{af_a+bf_a+\nu Y(bY),f_b}{
        \vliiin{\lr+}{}{af_a+bf_a+\nu Y(bY),af_b+bf_b+\nu Y(aY)}{
          \vliiin{\lr+}{}{af_a,af_b+bf_b+\nu Y(aY)}{
            \vlin{\kk a}{}{af_a,af_b}{
              \vlin{}{\bullet}{f_a,f_b}{
                \vlhy{\vdots}
              }
            }
          }{
            \vlin{\dis}{}{af_a,bf_b}{
              \vlhy{}
            }
          }{
            \vliq{}{}{af_a,\nu Y(aY)}{
              \vlhy{}
            }
          }
        }{
          \vliiin{\lr+}{}{bf_a,af_b+bf_b+\nu Y(aY)}{
            \vlin{\dis}{}{bf_a,af_b}{
              \vlhy{}
            }
          }{
            \vlin{\kk b}{}{bf_a,bf_b}{
              \vlin{}{\bullet}{f_a,f_b}{
                \vlhy{\vdots}
              }
            }
          }{
            \vlin{\lr\nu}{}{bf_a,\nu Y(aY)}{
              \vlin{\dis}{}{bf_a,a\nu Y(aY)}{
                \vlhy{}
              }
            }
          }
        }{
          \vliiin{\lr+}{}{\nu Y(bY),af_b+bf_b+\nu Y(aY)}{
            \vlin{\lr\nu}{}{\nu Y(bY),af_b}{
              \vlin{\dis}{}{b\nu Y(bY),af_b}{
                \vlhy{}
              }
            }
          }{
            \vliq{}{}{\nu Y(bY),bf_b}{
              \vlhy{}
            }
          }{
            \vlin{\lr\nu}{}{\nu Y(bY),\nu Y(aY)}{
              \vlin{\lr\nu}{}{b\nu Y(bY),\nu Y(aY)}{
                \vlin{\dis}{}{b\nu Y(bY),a\nu Y(aY)}{
                  \vlhy{}
                }
              }
            }
          }
        }
      }
    }
  }
}$}}
    \caption{A $\CRLLLang$ proof that no $\omega$-word over $\{a,b\}$ can have both finitely many $a$s and finitely many $b$s. All formulas occur to the LHS of the sequent arrow $\seqar$, which we have omitted to save space.}
    \label{fig:fa-cap-fb-empty}
\end{figure}

\section{Soundness and completeness}
\label{sec:adequacy}
We are now in a position to state the main result of this work:

\begin{theorem}
[Adequacy]
\label{thm:adequacy-crlll-lang}
$\CRLLLang\proves \Gamma\seqar \Delta \iff  \bigcap\limits_{e \in \Gamma}\lang e \subseteq \bigcup\limits_{f\in \Delta}\lang f$, for $\Gamma,\Delta$ containing only guarded expressions.
\end{theorem}
Both directions are ultimately established by reduction to the adequacy of evaluation games, \cref{lem:eval}, though completeness will require us to build up further game theoretic machinery first.

\begin{remark}
    [Restriction to guarded sequents]
    The restriction to guarded expressions in our main theorem above is harmless in the sense that each expression is equivalent to a guarded one, over $\Lang$, cf.~\cref{prop:omega-reg-have-munu-exps} (see also, e.g., \cite{BruseFL15:guarded-transformations} for the related setting of the modal $\mu$-calculus). 
    For instance $\mu XX$ is equivalent to $\mu X(aX)$ and $\nu XX$ is equivalent to $\nu X \left(\sum\limits_{a \in \Alphabet} aX\right)$.
Dealing with unguarded expressions introduces complications for proof search strategy, in particular requiring tedious loop checks, (see, e.g., \cite{FriLan11:mu-calc-off-guard}, again in the setting of the modal $\mu$-calculus) that arguably detract from the deeper game theoretic machinery underlying completeness proofs in cyclic proof theory.
We should emphasise that we believe that our system $\CRLLLang$ is sound and complete over \emph{arbitrary} expressions, though such a result is beyond the scope of this work.
Let us point out that, in previous work \cite{DD24a,DasDe24:rla-preprint}, guardedness was \emph{necessary} for completeness for a subsystem $\nu\CRLA$, without $\top,\cap$ and with \emph{exactly} one formula on the left.
$\CRLLLang$ does not fall victim to the same counterexamples, due to its symmetric nature.
\end{remark}

\subsection{Soundness}
\label{sec:soundness-cyclic}

In this subsection we prove soundness of $\CRLLLang$ for $\Lang$.
\begin{proof}
[Proof of $\implies$ direction of \cref{thm:adequacy-crlll-lang}]
    Let $P$ be a $\LRLLLanghat$ preproof of $\Gamma\seqar \Delta$ and suppose $w\in \lang e$ for all $e \in \Gamma$ but $w \notin \lang f $ for all $f \in \Delta$.
    We shall show that $P$ is not progressing.

    Let $\mathfrak e $ be a positional $\Eloise$ strategy that is winning from each $(w,e)$ for $e \in \Gamma$ and $\strat a $ a positional $\Abelard$ strategy winning from each $(w,f)$ for $f\in \Delta$.\footnote{For instance, it suffices to take an $\Eloise$ strategy winning from $(w,\bigcap\Gamma)$ and an $\Abelard$ strategy winning from $(w,\sum \Delta)$. 
    Note that there always even exists a \emph{universal} positional winning strategy for each player, one that wins from all winning positions, but we are not using this fact.}
     $\mathfrak e$ and $\mathfrak a $ induce a unique infinite branch $B_{\strat e ,\strat a}$, by always making choices in the branch construction consistent with these two strategies.
    In more detail, we construct $B_{\strat e ,\strat a} = (\Gamma_i \seqar \Delta_i)_{i<\omega}$ and $(w_i)_{i<\omega}$, with each $w_i \in \Alphabet^\omega$, by setting:
    \begin{itemize}
        \item $\Gamma_0 \seqar \Delta_0$ is $\Gamma\seqar \Delta$ and $w_0 $ is $w$.
        \item When $w_i = av$ and $\Gamma_i \seqar \Delta_i$ concludes a $\kk a $ step we set $w_{i+1} \df v$ (and $\Gamma_{i+1} \seqar \Delta_{i+1}$ to be the only premiss).
        \item When $w_i = av$ and $\Gamma_i\seqar \Delta_i$ concludes a $\exh$ step we set $w_{i+1} \df v$ and $\Gamma_{i+1}\seqar \Delta_{i+1}$ to be the $a$-premiss.
        \item When $\Gamma_i \seqar \Delta_i$ concludes a $\lr +$ step with principal formula $g_0+g_1$, if $\strat e$ plays $(w_i,g_0 + g_1) \to (w_i,g_j)$ then $B_{\strat e ,\strat a}$ follows the $g_j$ branch and $w_{i+1} \df w_i$.
        \item When $\Gamma_i \seqar \Delta_i$ concludes a $\rr \cap$ step with principal formula $g_0\cap g_1$, if $\strat a$ plays $(w_i,g_0 \cap g_1) \to (w_i,g_j)$ then $B_{\strat e ,\strat a}$ follows the $g_j$ branch and $w_{i+1} \df w_i$.
        \item In all other cases $B_{\strat e, \strat a}$ follows the only premiss and $w_{i+1}=w_i$. 
        Note in particular that no $\Gamma_i\seqar \Delta_i$ can conclude an initial step by local soundness, \cref{prop:local-soundness}.
    \end{itemize}
        Recall from \cref{obs:guarded-preproofs-no-stable-traces} that, by guardedness, no trace in $P$ is ultimately stable.
    \anupam{actually I don't think we are relying on this! we just need that any infinitely often principle on the LHS/RHS is a play of e/a resp. Thus this proof should hold also for unguarded sequents.}
    \anupam{actually, what about $\mu XX$ and $\nu XX$ that unfold to themselves?}
    By inspection of the rules and definition of $B_{\strat e , \strat a }$ we thus have:
    \begin{itemize}
        \item Every LHS trace along $B_{\strat e,\strat a}$, restricted to principal formulas, forms the right components of a play of $\strat e$ from $(w,e)$, for some $e\in \Gamma$.
        \item Every RHS trace along $B_{\strat e,\strat a}$, restricted to principal formulas, forms the right components of a play of $\strat a$ from $(w,f)$, for some $f\in \Delta$.
    \end{itemize}
    Now, $B_{\strat e,\strat a}$ cannot have a progressing trace on the LHS as $\strat e$ is winning for $\Eloise$, nor on the RHS as $\strat a $ is winning for $\Abelard$.
    Thus $P$ is not progressing.
\end{proof}

\begin{remark}
    [The need for positionality]
     Note that, in the proof above, we really must exploit positionality of $\strat e $ and $\strat a$ during the construction of $B_{\strat e, \strat a}$. A formula occurrence may have multiple trace histories, as the graph of immediate ancestry is not necessarily a tree when cedents are sets.
     Thus the inductive construction of $\Gamma_i \seqar \Delta_i$, in particular at a $\lr +$ or $\rr \cap$ step, is only well-defined for $\strat e$ and $\strat a $ positional.
\end{remark}

\subsection{The proof search game and determinacy}
In order to prove completeness of $\CRLLLang$ for $\Lang$, we rely on further game determinacy principles to organise bottom-up proof search appropriately.

\begin{definition}
    [Proof search game]
    The \emph{proof search game} (for $\LRLLLanghat$) is a two-player game played between Prover $(\prover)$, whose positions are inference steps of $\LRLLLanghat$, and Denier $(\denier)$, whose positions are sequents of $\LRLLLanghat$.
A \defname{play} of the game starts from a particular sequent:
at each turn, $\prover$ chooses an inference step with the current sequent as conclusion, and $\denier$ chooses a premiss of that step; the process repeats from this sequent as long as possible.

An infinite play of the game is \defname{won} by $\prover$ (aka \defname{lost} by $\denier$) if the branch constructed has a progressing trace; otherwise it is won by $\denier$ (aka lost by $\prover$). In the case of deadlock, the player with no valid move loses.\footnote{Technically, no position is a deadlock for $\prover$, as $\exh$ can even be applied to the empty sequent.}
\end{definition}

Note that any $\LRLLLanghat$ preproof can have only finitely many distinct sequents. The only formulas that may occur in a $\LRLLLanghat$ preproof $P$ of a sequent $\vec e \seqar \vec f$ belong to either some $\fl {e_i}$ or $\fl {f_j}$, of which there are finitely many, cf.~\cref{prop:fl-props}.
As sequents are sets of formulae, there are thus only finitely many sequents occurring in $P$.
As a result the proof search game for $\LRLLLanghat$ is \emph{finite state}: it has only finitely many positions.
From here it is not hard to see that the set of progressing branches from $\vec e\seqar \vec f$ forms an $\omega$-regular language over the (finite) alphabet of possible sequents, similarly to the progress checking automaton $\mathbf B_P$ from \cref{proof-checking}, only now independent of the particular preproof $P$, depending only on the end-sequent $\vec e\seqar \vec f$.
Consequently we have:

\begin{proposition}
[B\"uchi-Landweber]
    The proof search game from any sequent is finite memory determined.
\end{proposition}
Here `finite memory', in particular, means that the strategy needs only store a bounded amount of information at any time (see, e.g., \cite[Section 4]{PerPin:inf-word-aut-book} for more on $\omega$-regular games).
It is not hard to see that any finite memory $\prover$ strategy is just a regular preproof (where the finite memory corresponds to multiple, yet finite, occurrences of the same sequent).
A similar analysis applies to $\denier$ strategies.\footnote{The fact that a player having a winning strategy means they have a \emph{regular} winning strategy can also be seen as a consequence of Rabin's basis theorem, as the set of strategies (for either player) forms an $\omega$-regular tree language.}
Moreover if the strategy is winning for $\prover$ then the corresponding preproof is progressing.
Thus we have:
\begin{corollary}
    \label{cor:reg-prf-or-denier-win-strat}
    Any sequent is either $\CRLLLang$-provable or has a regular $\denier$ winning strategy.
\end{corollary}


\subsection{A proof search strategy}
Before giving our completeness argument, let us describe a basic validity-preserving proof search algorithm, which shall serve as a `canonical' $\prover$ strategy in the proof search game (at least from guarded sequents).

\begin{definition}
    [Proof search strategy]
    \label{def:prf-search-strat}
    The $\prover $ strategy $\strat P$ in the proof search game is defined as follows, bottom-up:
    \begin{itemize}
        \item $\strat P$ applies left and right logical rules maximally.\anupam{make this also `fair' and I think we need not assume guardedness. Also need a weakening lemma, e.g. for proof of $\nu X X \seqar \sum\limits_{a\in \Alphabet}a\top$}
        
        \item At a sequent $\Gamma, ae,bf\seqar \Delta$ with $a\neq b$, $\strat P$ weakens $\Gamma$ and $\Delta $ then applies $\dis$ to finish the proof.
        \item At a sequent $a\Gamma \seqar \{b\Delta_b\}_{a\in \Alphabet}$, with $\Gamma \neq \emptyset$, $\strat P$ weakens all $b\Delta_b$ for $b\neq a$ and applies $\kk a $ to give the premiss $\Gamma \seqar \Delta_a$.
        \item At a sequent $\seqar \{a\Delta_a\}_{a\in \Alphabet}$ $\strat P$ applies $\exh$.
    \end{itemize}
\end{definition}

Note that $\strat P$ is indeed a well defined total $\prover$ strategy: no logical rule applies just when every formula in the sequent has form $ae$.\anupam{is it important that the strategy preserves nonemptiness of LHS?}
Moreover we have:
\begin{proposition}
[Validity preservation]
    $\strat P$ is validity-preserving. I.e.\ any play of $\strat P$ from a $\Lang$-valid sequent visits only $\Lang$-valid sequents.
\end{proposition}
\begin{proof}
[Proof sketch]
By inspection the logical rules, $\dis$ and $\exh$ are all invertible, i.e.\ they preserve validity bottom-up. 
For the third case, weakening before applying the $\kk a $ rule is justified by the fact that the ranges of $a\cdot$ and $b\cdot$ are disjoint in $\Lang$ when $a\neq b$, cf.~\eqref{eq:letters-freely-cogenerate-structure}.
\end{proof}

Note also that, when playing from sequents containing only guarded formulas, no play of $\strat P$ eventually only applies logical rules: the $\Alphabet$ rules must be applied infinitely often.\todo{could insert a proof sketch here}
This turns out to be critical for the countermodel construction below: for unguarded sequents some extra `loop checking' required to define an appropriate $\prover $ strategy.

\subsection{Completeness}

By \cref{cor:reg-prf-or-denier-win-strat} above
it suffices, for completeness, to expand each $\denier$ winning strategy from a sequent into a countermodel:

\begin{lemma}
[Countermodel from $\denier$ winning strategy]
\label{lem:countermodel-from-D-win-strat}
    If $\denier$ has a winning strategy from a sequent $\Gamma\seqar \Delta$, for $\Gamma,\Delta$ containing only guarded expressions, there is some $w \in \bigcap\limits_{e\in \Gamma} \lang e$ with $w\notin \bigcup\limits_{f\in \Delta} \lang f $.
\end{lemma}
\begin{proof}
Let $\strat D$ be a $\denier$ winning strategy from $\Gamma \seqar \Delta$ and $\strat P$ the $\prover$ strategy from \cref{def:prf-search-strat}.
 Play $\strat P$ against $\strat D$ from $\Gamma \seqar \Delta$ to obtain a branch $B$, which must be infinite by assumption that $\strat D$ is $\denier$-winning and since $\strat P$ is total.
    By guardedness, the play infinitely often applies an $\Alphabet$ step.
    For $i<\omega$ write $\infrule_i$ for the $i$\textsuperscript{th} $\Alphabet$ step along $B$, bottom-up. 
    We define the infinite word $w_B \df (a_i)_{i<\omega}$ as follows:
     \begin{itemize}
        \item if $\infrule_i$ is $\kk a $ then $a_i \df a$;
        \item if $\infrule_i $ is $\exh$ and $B$ follows the $a$-premiss then $a_i \df a$.
        \item ($\infrule_i$ cannot be $\dis$ as $B$ is infinite).
    \end{itemize}

     We will show that $w_B\in \lang e$ for each $e \in \Gamma$ but $w_B \notin \lang f $ for each $f \in \Delta$.
    First let us note also that, by guardedness, $B$ has no ultimately stable traces, cf.~\cref{obs:guarded-preproofs-no-stable-traces}. 
    Thus, since $\strat D $ is $\denier$-winning: 
    \begin{enumerate}
        \item\label{item:lhs-trace-nu} each LHS trace along $B$ is a $\nu$-trace; and,
        \item\label{item:rhs-trace-mu} each RHS trace along $B$ is a $\mu$-trace.
    \end{enumerate}
Now we have:
    \begin{itemize}
        \item $w_B \in \lang e$ for each $e \in \Gamma$. 
        By inspection of the definitions of $\strat P$ and $w_B$, note that the LHSs of $B$ uniquely determine an $\Eloise$ strategy $\strat e$ in the Evaluation Game from $(w_B,e)$, by restricting the ancestry graph from $e$ to principal formulas and conclusions of $\Alphabet$ steps.
        It is important here that $B$ is free of $\lr \wk$ steps (as $\strat P$ never applies them) so that $\strat e $ is total.
        By \eqref{item:lhs-trace-nu}, $\strat e$ is winning for $\Eloise$, so we conclude by adequacy of the Evaluation Game, \cref{lem:eval}.
        \item $w_B \notin \lang f$ for each $f\in \Delta$. 
        The RHSs of $B$ uniquely determine an Abelard strategy $\strat a$ in the evaluation game from $(w_B,f)$, by restricting the ancestry graph from $f$ to principal occurrences and conclusions of $\Alphabet$ steps.
        Note here that, again for totality of $\strat a$, the only $\rr \wk$ steps are applied at a sequent $a\Gamma \seqar \{b\Delta_b\}_{b \in \Alphabet}$ on formulas $bf$ for $b\neq a$.
        By definition any play of $\strat a$ up to such a weakened formula reaches a position $(aw,bf)$ with $b\neq a$ and so immediately $\Abelard$ wins anyway.
        So $\strat a $ is totally defined and any infinite play of it must be winning for $\Abelard$ by \eqref{item:rhs-trace-mu}.
        Again we conclude by adequacy of the Evaluation Game, \cref{lem:eval}. \qedhere
    \end{itemize}
\end{proof}

\begin{remark}
    [`Finite' countermodels]
Note that the countermodel generated above is actually rational when $\strat D $ is finite memory, i.e.\ $w_B$ is an ultimately periodic word. 
We may thus construe it as a finite countermodel, by representing it as a graph.
This is unsurprising in light of known basis theorems in the theory of $\omega$-regularity.
\end{remark}

We can now finally conclude our proof of the adequacy of $\CRLLLang$ for $\Lang$:
\begin{proof}
[Proof of $\impliedby$ direction of \cref{thm:adequacy-crlll-lang}]
    By contraposition. Suppose $\CRLLLang \not \proves \Gamma\seqar \Delta$, in which case there is a $\denier$ winning strategy from $\Gamma\seqar \Delta$ in the proof search game, by \cref{cor:reg-prf-or-denier-win-strat}.
    Thus by \cref{lem:countermodel-from-D-win-strat}, we have that $\bigcap\limits_{e\in \Gamma}\lang e \not \subseteq \bigcup\limits_{f \in \Delta}\lang f$.
\end{proof}

\section{Conclusions and further remarks}
\label{sec:concs}
In this work we studied an algebraic notation for alternating parity automata (APAs), in the form of right-linear lattice (RLL) expressions.
We designed a cyclic proof system $\CRLLLang$ manipulating RLL expressions, reasoning about inclusions between the $\omega$-languages they denote (thus also emptiness, universality and equivalence). 
Along the way we developed several game theoretic techniques, inspired by the cyclic proof theory tradition, for organising metalogical reasoning, in particular for our main soundness and completeness result, \cref{thm:adequacy-crlll-lang}.

\medskip

As $\CRLLLang$ is an analytic system,\footnote{While $\CRLLLang$ does not have the subformula property per se, the finitude of Fischer-Ladner closures means that only finitely many formulas occur in a proof, even linearly many in the size of the end-sequent, cf.~\cref{prop:fl-props}.} it is amenable to proof search and implementation, potentially offering new algorithms for deciding APA inclusion (and emptiness, universality, equivalence). 
Deciding inclusion for $\omega$-regular languages is of fundamental interest in \emph{model checking} (see, e.g., \cite{VW94,GPVW96,GO01,Holzmann11}).
Typical approaches only handle B\"uchi automata natively, so at least one novelty offered by $\CRLLLang$ is a new way to deal with expressions notating APAs directly, rather than their encodings by B\"uchi automata.
A feature of proof search based algorithms in general is that they admit a notion of \emph{certificate}: if we find a proof $P$ of $e\seqar f$, we can easily convince a third party (`checker') that $\lang e \subseteq \lang f$ by communicating $P$ as evidence, rather than requiring them to reprove it outright.

RLL expressions do not include a native complement operation, unlike the linear-time $\mu$-calculus or monadic second-order logic. 
However the closure of $\omega$-regular languages under complement implies that we can \emph{define} complementary expressions $e^c$ for each expression $e$, s.t.\ $\CRLLLang$ proves both sequents $\seqar e,e^c$ and $e,e^c \seqar $.
In fact, the duality of our syntax means that $e^c$ is easy to define, by straightforward structural induction on $e$. 
In particular $(\mu X e)^c \df \nu X e^c$ (with $X^c \df X$), highlighting the duality of $\mu $ and $\nu$; proofs of the aforementioned sequents are guaranteed by completeness of $\CRLLLang$.

In parallel we have proposed an axiomatisation $\RLLLang$ for the RLL theory of $\omega$-regular languages, in particular proving its soundness and completeness wrt $\omega$-languages.
Here the derivation of complements plays a crucial role in the completeness argument, in particular as it implies that the initial term model of $\RLLLang$ forms a Boolean Algebra.

\bibliographystyle{alpha}
\bibliography{biblio}

\end{document}